\documentclass[a4paper,aps,amssymb,amsmath,amsfonts,superscriptaddress,
twocolumn]{revtex4-2}
\pdfoutput=1

\usepackage{physics}
\usepackage{graphicx}
\usepackage{bm,bbm}
\usepackage{epstopdf}
\usepackage{amsthm}
\usepackage{amsmath}
\usepackage{color}
\usepackage{amsopn}
\usepackage{amssymb}
\usepackage{hyperref}
\usepackage{subcaption}
\usepackage{enumerate}
\usepackage{todonotes}
\usepackage{ dsfont }
\usepackage[T1]{fontenc}
\usepackage[english]{babel}
\usepackage[utf8]{inputenc}
\usepackage{mathtools}
\usepackage{multirow}
\usepackage{adjustbox}
\usepackage[utf8]{inputenc}

\newtheorem{lemma}{Lemma}

\newtheorem{theorem}{Theorem}

\newtheorem{corollary}{Corollary}

\newcommand{\N}{\ensuremath{\mathbb{N}}}

\newcommand{\C}{\ensuremath{\mathbb{C}}}
\newcommand{\proj}[1]{\ensuremath{\ketbra{#1}{#1}}}
\newcommand{\floor}[1]{\ensuremath{\lfloor #1 \rfloor}}
\newcommand{\Id}{{\rm 1\hspace{-0.9mm}l}}
\newcommand{\LL}{\mathcal{L}}

\newcommand{\RR}{\mathcal{R}}
\newcommand{\PP}{\mathcal{P}}
\newcommand{\CC}{\mathcal{C}}
\newcommand{\ie}{\emph{i.e.} }

\newcommand{\QQ}{\mathcal{Q}}

\newcommand{\TT}{\mathcal{T}}
\renewcommand{\SS}{\mathcal{S}}

\newcommand{\II}{\mathcal{I}}

\newcommand{\HH}{\mathcal{H}}
\newcommand{\XX}{\mathcal{X}}
\newcommand{\YY}{\mathcal{Y}}
\newcommand{\ZZ}{\mathcal{Z}}

\newcommand{\ketv}[1]{\ensuremath{\left|\left.{#1} \right\rangle \! 
\right\rangle}}
\newcommand{\brav}[1]{\ensuremath{\left\langle \! \left\langle {#1} 
\right.\right|}}
\newcommand{\projv}[1]{\ketv{{#1}} \!\! \brav{{#1}}}

\begin{document}
	
\title{Storage and retrieval of von Neumann
	measurements}
\author{Paulina Lewandowska}
\email{plewandowska@iitis.pl}
\affiliation{Institute of Theoretical and Applied Informatics, Polish Academy
	of Sciences, Ba{\l}tycka 5, 44-100 Gliwice, Poland}
\author{Ryszard Kukulski}
\affiliation{Institute of Theoretical and Applied Informatics, Polish Academy
	of Sciences, Ba{\l}tycka 5, 44-100 Gliwice, Poland}
\author{\L ukasz Pawela}
\affiliation{Institute of Theoretical and Applied Informatics, Polish Academy
	of Sciences, Ba{\l}tycka 5, 44-100 Gliwice, Poland}
\author{Zbigniew Pucha\l a}
\affiliation{Institute of Theoretical and Applied Informatics, Polish Academy
	of Sciences, Ba{\l}tycka 5, 44-100 Gliwice, Poland}

\begin{abstract}
	This work examines the problem of learning an unknown von 
	Neumann measurement of dimension $d$ from a finite number of copies. To obtain a faithful 
	approximation of the given measurement we are 
	allowed to use it $N$ times. Our main goal is to estimate the asymptotic 
	behavior of the maximum 
	value of the average fidelity function $F_d$ for a general $N \rightarrow 
	1$ learning scheme. We show that $F_d = 1 - 
	\Theta\left(\frac{1}{N^2}\right)$ for arbitrary but fixed dimension $d$. 
	In addition to that, we compared various learning 
	schemes for $d=2$. We observed that the learning scheme based on 
	deterministic port-based teleportation is asymptotically optimal but 
	performs poorly for low $N$. In particular, we discovered a  parallel 
	learning 
	scheme, which despite its lack of 
	asymptotic optimality, provides a high value of the fidelity for 
	low values of $N$ and uses only two-qubit entangled memory states.

\end{abstract}
\maketitle
\section{Introduction}
In the general approach of storage and retrieval (SAR) of quantum operations, 
we want 
to approximate a given, unknown operation, which we
were able to perform $N$ times experimentally. Such a scheme is called $N 
\rightarrow 1$ \emph{learning scheme.} 
This strategy usually consists of preparing some initial quantum state, applying the unknown 
operation $N$ times,
which allows us to store the unknown operation for later use, and finally, a 
retrieval operation that
applies an approximation of the black box on some arbitrary quantum state. 
Additionally, each application
of the operation contained within the black box can be followed by 
some arbitrary
processing operations. If that is the case, the optimal strategy should also 
contain their
description. The scheme is optimal when it achieves the highest possible 
fidelity of the
approximation~\cite{raginsky2001fidelity, belavkin2005operational}.

The obstacle standing in the way of achieving unit fidelity lies in the 
no-cloning 
theorem~\cite{wootters1982single} and even further, the no-programming 
theorem~\cite{nielsen1997programmable}. It states that a general processor, 
which performs a program based on some input 
state is not possible. There is no doubt that programmable devices would 
represent an instrumental 
piece of quantum technology. Hence, their approximate realizations are of 
common interest~\cite{buvzek1996quantum,
	hillery2002probabilistic, yang2020optimal, 
	gschwendtner2021programmability, peres2002unspeakable}. 
Nevertheless, one should keep in mind the difference 
between programming and learning. A learning protocol can be used to generate a 
program, but the program does not have to be generated by learning. Although,  
in some cases, the performances of programming and learning are equal 
\cite{yang2020optimal}.

The seminal work in this field was the paper by Bisio and 
Chiribella~\cite{bisio2010optimal}. It was devoted
to learning an unknown unitary transformation. Therein, the authors focused on 
storing the unitary
operation in a quantum memory while having limited resources. They 
proved that unitary
operations could be learned optimally in the parallel scheme, meaning there 
is no additional
processing after using the unknown unitary transformation. Hence, all 
the required uses of
the black box can be performed in parallel. They also provide an upper bound on 
the fidelity of
$N \rightarrow1$ learning scheme that is equal
$1-\Theta\left(\frac{1}{N^2}\right)$. A 
probabilistic version of
SAR (PSAR) problem was also considered in \cite{sedlak2019optimal, 
sedlak2020probabilistic}. There,
they showed the optimal success probability of $N\rightarrow 1 $ PSAR of 
unitary channels on
$d$-dimensional quantum systems is equal to $N/(N-1+d^2)$.

It is worth mentioning that the SAR of unitary operations is also closely 
related to port-based teleportation (PBT) \cite{ishizaka2008asymptotic, 
ishizaka2009quantum}. One may take the PBT protocol to define a learning scheme. In 
particular, the entanglement fidelity 
of the deterministic PBT protocol matches the fidelity of learning unitary 
operations~\cite{bisio2010optimal, christandl2021asymptotic}. Similarly, the 
probability of successful teleportation is equal to the performance of PSAR 
\cite{sedlak2019optimal}. 
Finally, the task of learning unknown operations 
fits into the paradigm of quantum machine learning, when both the 
data and the algorithms are quantum~\cite{sasaki2002quantum, sentis2012quantum, 
	dunjko2016quantum,
	monras2017inductive, alvarez2017supervised, amin2018quantum, 
	sentis2019unsupervised}.

Subsequent works build upon these results but focus on different classes of 
operations, for example,
the von Neumann measurements~\cite{bisio2011quantum}. In contrast to previous 
works, they showed
that, in general, the optimal algorithm for quantum measurement learning cannot 
be parallel and found
the optimal learning algorithm for arbitrary von Neumann measurements for the 
case $1 \rightarrow 1$
and $2\rightarrow 1$. Nevertheless, a general optimal scheme $N \rightarrow 1 $ 
of measurement
learning still remains an open problem, even for low-dimensional quantum 
systems. Hence, despite some partial results, the investigation of SAR for von 
Neumann measurements is still an open question.

In this work, we address the unsolved problem of learning an unknown von 
Neumann measurement
defined in \cite{bisio2011quantum}. We focus on $N \rightarrow 1$  learning  scheme of von Neumann measurements. We investigate a 
value of the average fidelity 
in the asymptotic regime. By using the deterministic PBT protocol we state a 
lower bound which behaves as $1 - \Theta\left(\frac{1}{N^2}\right)$. Moreover, 
we provide an upper bound for the average fidelity function, which matches the 
lower bound and hence provides a solution to the problem. 

Additionally, we compare different learning 
schemes for the qubit case. Although, the learning scheme based on PBT is 
asymptotically optimal, it can be outperformed for low values of $N$. To show this,  we introduce 
a scheme, which we call  \emph{pretty good learning scheme} (PGLS). 
This scheme is a particular case of a parallel learning scheme 
which uses only two-qubit entangled memory states. The fidelity function 
calculated for the pretty good learning
scheme is uniform over all qubit von Neumann measurements and behaves as 
$1-\Theta\left(\frac{1}{N}\right)$.  

This paper is organized as follows. In Section~\ref{sec:formulation} we 
formulate the problem of von
Neumann measurement learning. In Section~\ref{sec:notation} we introduce 
necessary mathematical
concepts. Our main result is then presented in Section~\ref{sec:setup}
(Theorem~\ref{theorem}). To prove this theorem, we first address the case of 
lower bound (Section~\ref{fidelitylower}), and subsequently,  upper bound 
(Section~\ref{fidelityupper}). In Section~\ref{sec:qubit} we compare the 
performance of different learning schemes for a qubit case. In particular, we 
introduce the pretty good learning scheme and present numerical results about the most efficient parallel and adaptive learning schemes. 
Finally,
Section~\ref{sec:conclusion} concludes the article with a summary of the main 
results. In the Appendix, we
provide technical details of proofs.

\section{Problem formulation}\label{sec:formulation}

This section presents the formulation of the problem of learning an 
unknown von Neumann
measurement. We provide an overview of a learning scheme in 
Fig.~\ref{fig:schema},  along with its
description in Subsection~\ref{sec:setup}.

\subsection{Mathematical framework}\label{sec:notation}

Let us introduce the following notation. Consider a $d$-dimensional complex 
Euclidean space $\C^d$ and
denote it by $\HH_d$. Let $\mathrm{M}(\HH_{d_1}, \HH_{d_2})$ be the set of all 
matrices of dimension
$d_1 \times d_2 $. As a shorthand we put $\mathrm{M}(\HH_d) \coloneqq 
\mathrm{M}(\HH_d, \HH_d)$. The
set of quantum states defined on space $\HH_d$, that is the set of positive 
semidefinite operators having unit trace,
will be denoted by $\Omega(\HH_d)$. We will also need a linear mapping 
transforming
$\mathrm{M}(\HH_{d_1}) $ into $\mathrm{M}(\HH_{d_2}) $ as $\mathcal{T}: 
\mathrm{M}(\HH_{d_1})
\mapsto \mathrm{M}(\HH_{d_2}).$ There exists a bijection between introduced 
linear mappings
$\mathcal{T}$ and set of matrices $\mathrm{M}(\HH_{d_1  d_2}) $, known as the 
Choi-Jamio{\l}kowski
isomorphism~\cite{choi1975completely, jamiolkowski1972linear}. Its explicit 
form is $\text{T} =
\sum_{i,j=0}^{d_1-1} \TT(\ketbra{i}{j}) \otimes \ketbra{i}{j} $. We will denote 
linear mappings with calligraphic font $\LL, \SS, \mathcal{T}$ etc., whereas 
the 
corresponding Choi-Jamio{\l}kowski matrices as plain symbols: $L, S, T$ etc. 
Moreover, we introduce the vectorization operation of a matrix $X \in 
\mathrm{M}(\HH_{d_1}, 
\HH_{d_2})$, defined by 
$\ketv{X} 
\coloneqq
\sum_{i=0}^{d_2-1} \left(X\ket{i}\right) \otimes \ket{i}$.

A general quantum measurement (POVM) $\mathcal{Q}$ can be viewed as a set of 
positive semidefinite
operators $ \QQ = \{ Q_i \}_i$ such that $\sum_i Q_i = \Id$. These operators 
are usually called 
effects. The von Neumman measurements, $\PP_U$,
are a special subclass of measurements whose all effects are rank-one 
projections given by $\PP_U =
\{P_{U,i}\}_{i=0}^{d-1} = \{U\ketbra{i}{i}U^\dagger\}_{i=0}^{d-1}$ for some 
unitary matrix $U \in
\mathrm{M}(\HH_d)$.

Quantum channels are completely positive and trace preserving (CPTP) linear 
maps. Generally, $\CC$
is a quantum channel which maps $\mathrm{M}(\HH^{(in)})$ to 
$\mathrm{M}(\HH^{(out)})$ if its 
Choi-Jamio{\l}kowski operator $C$ is a positive semidefinite and 
$\tr_{\HH^{(out)}} (C)
= \Id$, where $\tr_{\HH^{(out)}}$ denotes a partial trace over the output 
system $\HH^{(out)}$.
Given a von Neumann measurement $\PP_U$, it can be seen as a 
measure-and-prepare quantum channel
$\PP_U(\rho) = \sum_{i} \tr\left( P_{U,i} \rho \right)\ketbra{i}{i}$, $\rho \in 
\Omega(\HH_d)$. The
Choi matrix of $\PP_U$ is $ P_U = \sum_{i} \ketbra{i}{i} \otimes 
\overline{P_{U,i}}, $ which will be
utilized throughout this work. Finally, we will use the notation $\Phi_U$ to 
indicate
unitary channel given by $ \Phi_{U}(X) = U X U^\dagger $ and the shortcut 
$\II_d \coloneqq
\Phi_{\Id_d}$ for the identity channel.

\subsection{Learning setup}\label{sec:setup}

\begin{figure}[!ht]
\centering
	\includegraphics[scale=0.9]{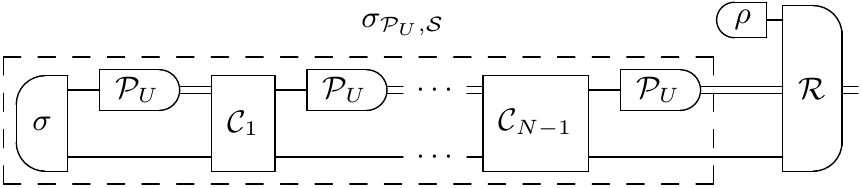}
	\caption{Schematic representations of the setup for learning of von Neumann
		measurements $\PP_U$ in the $N \rightarrow 1$ scheme.}\label{fig:schema}
\end{figure}
Imagine we are given a black box with the promise that it contains some von 
Neumann measurement,
$\PP_U$, parameterized by a unitary matrix $U$. The exact value of $U$ 
is unknown to us. We
are allowed to use the black box $N$ times. Our goal is to prepare some initial 
memory state
$\sigma$, some intermediate processing channels $\CC_1, \ldots, \CC_{N-1}$ and 
a  measurement $\RR$
such that we are able to approximate $\PP_U$ on an arbitrary state $\rho$. This 
approximation will
be denoted throughout this work as $\QQ_U$. We would like to point out that, 
generally, $\QQ_U$ will
\emph{not} be a von Neumann measurement.

The initial memory state $\sigma$ and entire sequence of processing channels 
$\{ \CC_i\}$ can be
viewed as \emph{storing} the unknown operation and will be denoted as $\SS$ 
whereas the measurement
$\RR$ we will call as \emph{retrieval}. During the storing stage, we apply 
$\SS$ on $N$ copies of
$\PP_U$. As a result, the initial memory state $\sigma$ is transferred to the 
memory state 
$\sigma_{\PP_U,\SS}$. After that, we measure an arbitrary quantum state $\rho$ 
and the memory state 
$\sigma_{\PP_U,\SS}$ by using $\RR$. Equivalently, we can say that during 
retrieval stage, we apply 
the measurement $\QQ_U$ on the state $\rho$. The entire learning scheme will be 
denoted by $\LL$ 
and considered as a triple $\LL = \left(\sigma, \{ \CC_i \}_{i=1}^{N-1}, \RR 
\right)$. We emphasize 
that the procedure allows us to use as much quantum memory as necessary.

As a measure of quality of approximating a von Neumann measurement $\PP_U=\{ 
P_{U,i}\}_i$ with a POVM $\QQ_U = \{ Q_{U,i}
\}_i$ we choose the fidelity function~\cite{raginsky2001fidelity}, which is 
defined as follows
\begin{equation}\label{fidelity}
	\mathcal{F}_d(\PP_U, \QQ_U) \coloneqq \frac{1}{d} \sum_i
	\tr(P_{U, i} Q_{U, i}),
\end{equation}
where $d$ is the dimension of the measured system. 
Note that in the case when $\PP_{U}$ is a von Neumann measurement we obtain the 
value of fidelity 
function $\mathcal{F}_d$ belongs to the interval $[0,1]$ and equals one if 
and 
only if $P_{U,i} = Q_{U,i}$ 
for all $i$. As there is no prior information about $\PP_U$
provided, we assume that $U$ is sampled from a distribution pertaining to the 
Haar measure.
Therefore, considering a von Neumann measurement $\PP_U$ and its approximation 
$\QQ_U$ we introduce
the \emph{average fidelity function}~\cite{bisio2016quantum} with respect to 
Haar measure as
\begin{equation}\label{average-fidelity}
	\mathcal{F}_d^{\text{avg}}(\LL) \coloneqq \int_U dU \mathcal{F}_d(\PP_U, 
	\QQ_U).
\end{equation}

Our main goal is to maximize $\mathcal{F}_d^\text{avg}$ over all possible 
learning schemes $\LL
= \left( \sigma, \{ \CC_i \}_{i=1}^{N-1}, \RR \right)$. We introduce the 
notation of the
maximum value of the average fidelity function
\begin{equation}\label{eq:fidelity}
	F_d \coloneqq \max_\LL \mathcal{F}_d^\text{avg}(\LL).
\end{equation}

In this work, we analyze the asymptotic behavior of $F_d$ with $N \to \infty$. 
Our main result can be summarized as the following theorem.

\begin{theorem}\label{theorem}
	Let $F_d$ be the maximum value of the average fidelity function, defined in
	Eq.~\eqref{eq:fidelity} for $N \rightarrow 1$ learning scheme of 
	von Neumann measurements. Then, for arbitrary but fixed dimension $d$ we 
	obtain
	\begin{equation}
		F_d = 1 - \Theta\left(\frac{1}{N^2}\right).
	\end{equation}
\end{theorem}

\section{Fidelity bounds}~\label{sec:fidelity-bounds}

This section provides a sketch of the proof of Theorem~\ref{theorem}, along with a general intuition behind our result. 
The full proofs are 
postponed to the Appendix~\ref{app:lower-bound},~\ref{app:upper-bound-2} and~\ref{app:upper} due to their technical nature.

\subsection{Lower bound}\label{fidelitylower}

The proof of the lower bound for $F_d$ is constructive. We 
will construct the learning scheme $\LL$ of von Neumann measurements, which 
achieves 
the scaling $ \mathcal{F}_d^{\text{avg}}(\LL) = 1 - 
\Theta\left(\frac{1}{N^2}\right)$.

\begin{figure}[h!]
	\centering\includegraphics{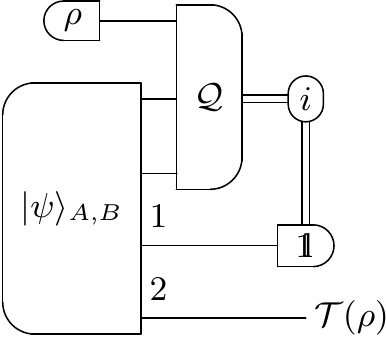}
	\caption{ Schematic representation of DPBT  for $N=2$. In this case, the 
	output label $i=2$ of the 
	measurement $\QQ$ determines the partial trace of the first system of the 
	remaining quantum state. \label{fig-dpbt}}
\end{figure}

The construction is based on deterministic port-based teleportation (DPBT) 
\cite{ishizaka2008asymptotic, ishizaka2009quantum, studzinski2017port, 
mozrzymas2018optimal, christandl2021asymptotic}. In this scheme (see 
Fig.~\ref{fig-dpbt}), two 
individuals -- Alice and Bob, share an entangled state $\proj{\psi}_{A,B} \in 
\Omega(\HH_{d^N}^{(A)} \otimes \HH_{d^N}^{(B)})$. Bob perceives his compound 
system $\HH_{d^N}^{(B)}$ as a tensor product of $N$ systems (ports) of the form 
$\HH_{d^N}^{(B)} = \HH_d^{(1)}\otimes\cdots\otimes\HH_d^{(N)}$.
 Their goal is to teleport an 
unknown state $\rho \in \HH_d^{(in)}$ from Alice to Bob in a way that this 
state appears in one 
of Bob's ports. 
To achieve this, Alice performs appropriate measurement $\QQ = \{Q_i\}_{i=1}^N$ 
on $\rho \otimes \proj{\psi}_{A}$, receives one of the labels $\{ 
1,\ldots, N \}$ and communicates this label to Bob. By using the label $i$, Bob 
chooses $i$-th system of his state $\proj{\psi}_{B}$ as the one which contains 
the state $\rho$. The output of this procedure can be 
written as
\begin{equation}
	\sum_{i=1}^N \tr_{\HH_d^{(in)} \otimes  \HH_{d^N}^{(A)} \otimes 
	\HH_{d^{N-1}}^{(\bar B_i)}} \left( (Q_i \otimes \Id_{d^N})(\rho \otimes 
	\proj{\psi}_{A,B}) \right),
\end{equation}
where $\HH_{d^{N-1}}^{(\bar B_i)} \coloneqq \HH_d^{(1)} \otimes \cdots \otimes 
\HH_d^{(i-1)} 
\otimes \HH_d^{(i+1)} \otimes \cdots \otimes \HH_d^{(N)} $. In short, the 
output of this procedure will be 
denoted as $\TT(\rho)$, where $\TT$ is a channel describing DPBT, depending on 
the choice of $\proj{\psi}_{A,B}$ and $\QQ$.
It is known~\cite{christandl2021asymptotic}, that the best teleportation 
procedure $\TT_0$ approximates $\II_d$ with the entanglement fidelity
\begin{equation}
	F_* \coloneqq \frac{1}{d^2}\tr \left(T_0 \projv{\Id_d}\right) = 1 - 
	\Theta\left(\frac{1}{N^2}\right).
\end{equation}

\begin{figure}[h!]
	\centering\includegraphics{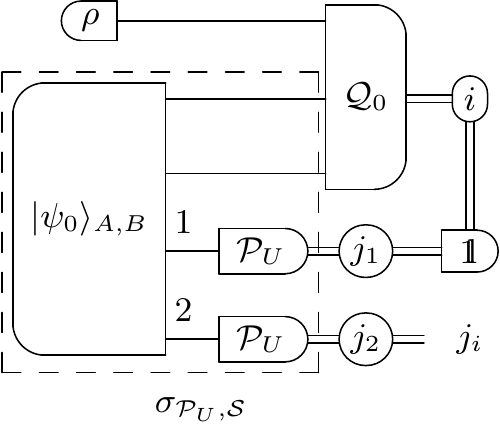}
	\caption{Schematic representation of the learning
		scheme for $N = 2$ based on DPBT. In this case, the label $i=2$ of 
		the measurement $\QQ_0$ indicates that the output of the learning 
		procedure equals $j_i = j_2$. \label{fig-dpbt-learning}}
\end{figure}
We can use DPBT to construct a learning scheme (see 
Fig.~\ref{fig-dpbt-learning}). Let $\proj{\psi_0}_{A,B}$ and $\QQ_0$ realize 
the optimal teleportation strategy $\TT_0$. We take $\sigma = 
\proj{\psi_0}_{A,B}$ 
as an initial memory state and consider a parallel 
learning scheme~\cite{bisio2010optimal} with $N$ copies of the von Neumann 
measurement $\PP_{U}$. The result of the storage is a memory state 
$\sigma_{\PP_U, S} 
= \left( \II_{d^N} \otimes \PP_U^{\otimes N} \right) (\sigma)$, which consists 
of the remaining quantum state $\tau$, and a tuple of measurements' results 
$(j_1,\ldots,j_N) \in \{0,\ldots,d-1\}^N$. The retrieval 
$\RR$ is a composition of a measurement $\QQ_0$ and classical 
postprocessing. In details, first, we 
apply $\QQ_0$ on $\rho \otimes \tau$ to obtain the label $i \in 
\{1,\ldots,N\}$. Second, we return the result $j_i$ as the output of $\RR$.
This procedure determines a learning scheme $\LL$ which achieves the average fidelity
\begin{equation}
\mathcal{F}_d^{\text{avg}}(\LL) = 1 - 
\Theta\left(\frac{1}{N^2}\right).
\end{equation}
We postponed the details of the proof to Appendix~\ref{app:lower-bound}.

\subsection{Upper bound}\label{fidelityupper}

To show the upper bound for $F_d$, we will construct a different 
learning scheme based on 
the learning of unitary maps. It will provide the desired inequality, at first 
for $d=2$, then for arbitrary $d$. 

\begin{lemma}\label{lem-upper-main} 
	For $d=2$ the maximum value of the average fidelity 
function defined in
	Eq.~\eqref{eq:fidelity} is upper bounded by
	\begin{equation}
		F_2 \le 1 - \Theta\left(\frac{1}{N^2}\right).
	\end{equation}
\end{lemma}
The complete proof of Lemma~\ref{lem-upper-main} is shown in 
Appendix~\ref{app:upper-bound-2}.  As in 
the previous section, here we will only sketch the key steps.

\begin{figure}[!h]
	\centering\includegraphics{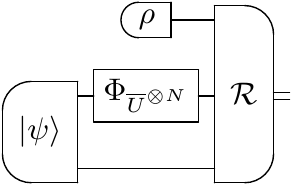}
	\caption{ Schematic representation of the setup, which we use to 
	calculate the upper 
	bound for $F_2$. In this 
		scenario, we are given $N$ copies of a unitary channel $\Phi_{\bar U}$ in 
		parallel. Our objective 
		is to approximate the von Neumann measurement $\PP_U$.	
		\label{upper-bound-3}}
\end{figure}
Let us consider a new learning scheme presented in Fig.~\ref{upper-bound-3} for 
$d=2$. In 
this scheme, we are 
given $N$ copies of unitary channel $\Phi_{\bar U}$, which we can use in 
parallel. We want to 
approximate the measurement $\PP_U$, but using the black box with the unitary 
channel 
$\Phi_{\bar U}$ inside. We will choose the appropriate initial memory state 
$\ket{\psi}$ and retrieval binary measurement $\RR = \{R_0, R_1\}$. We use the 
same measures of 
quality as before, namely 
$\mathcal{F}_2$ 
defined in Eq.~\eqref{fidelity} and $\mathcal{F}_2^{\text{avg}}$ defined in 
Eq.~\eqref{average-fidelity}. The goal is to maximize the value of the 
average fidelity 
function, which in this case, we will denote as $\widetilde{F_2}$. In the 
Appendix~\ref{upper:simplification-of-problem} we derived the formula for 
$\widetilde{F_2}$, which is given by
{\fontsize{9.5pt}{0}
\begin{equation}\label{eq-upper-main}
\begin{split}
	&\widetilde{F_2}=\\ & \max_{\RR, \proj{\psi}} \int_U dU  \sum_{i=0}^1
	\frac{\tr \left[ R_i \left(P_{U,i} \otimes \left(\Phi_{{\overline
				U}^{\otimes N}}
		\otimes \II \right)( \proj{\psi}) \right) \right]}{2}.
\end{split}
\end{equation}}
Calculating the value of $\widetilde{F_2}$ is the crux of the proof, because we 
managed to 
show that $F_2 \le \widetilde{F_2}$
(see Lemma~\ref{lemma-unitary-transformation} in 
Appendix~\ref{app:upper-bound-2}). We derived the thesis of 
Lemma~\ref{lem-upper-main}
by achieving the inequality $ \widetilde{F_2} \le 1 - 
\Theta\left(\frac{1}{N^2}\right).$

The proof of the upper bound for arbitrary $d$ relies on the qubit case. Let us 
take the optimal learning scheme $\LL$ such that it achieves 
$\mathcal{F}_d^{\text{avg}}(\LL) = F_d$ and satisfies 
the following commutation relation
\begin{equation}
	[L, \Id_d \otimes U \otimes (\Id_d \otimes \bar{U})]^{\otimes N}]=0
\end{equation}
for any unitary matrix $U \in \mathrm{M}(\HH_d)$. We can use $\LL$ to construct 
a  learning scheme $\LL'$ of qubit von Neumann measurements, 
such that it holds $\mathcal{F}_2^{\text{avg}}(\LL') \ge d F_d - (d-1)$. It 
directly implies that $F_d \le 1 - \Theta\left(\frac{1}{N^2}\right).$ We 
postponed the technical details to Appendix~\ref{app:upper}.

\begin{corollary}
	There  is no perfect learning scheme for von Neumann measurements, \ie 
	for any $N \in \N$ the value of $F_d < 1$.
\end{corollary} 

\section{Qubit case}\label{sec:qubit}

In this subsection we investigate more deeply the behavior of 
$\mathcal{F}_2^{\text{avg}}$ for different types of learning schemes $\LL$.

\subsection{Pretty good learning scheme}\label{qubit-pgls}

The first scheme which we analyze will be called the \emph{pretty good learning 
scheme}. Despite its lack of optimality, it provides a relatively high value for 
the average fidelity function asymptotically behaving as 
$\mathcal{F}_2^{\text{avg}}(\LL_{PGLS}) = 1 - 
\Theta\left(\frac1N\right)$. This scheme employs a simple storage strategy, 
which uses only two-qubit entangled memory states and the 
learning process is done in parallel. Moreover, the achieved value of the 
fidelity function is uniform over all qubit von Neumann measurements.

Let us consider a parallel learning scheme with $N$ 
copies of the von Neumann measurement $\PP_{U}$. A sketch of
our scheme is shown in Fig.~\ref{lower-bound-2} and here we present the 
algorithm describing the procedure:
\begin{figure}[h!]
	\centering\includegraphics{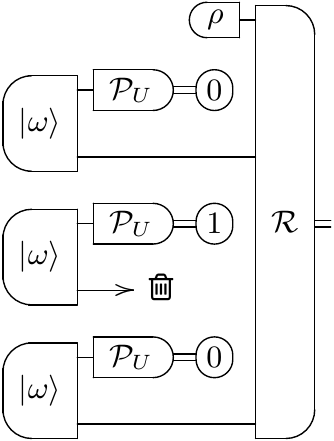}
	\caption{Schematic representations of the pretty good learning
		scheme for $N = 3$. In the learning process we obtained three labels:
		$0, 1, 0$. As labels ``$0$'' are in majority, we reject the label 
		``$1$'' and
		the associated quantum part.\label{lower-bound-2}}
\end{figure}
\begin{enumerate}
	\item We prepare the initial memory state $\sigma$ as a tensor product of
	$N$ maximally entangled states $\ket{\omega } \coloneqq \frac{1}{\sqrt{2}}
	\ketv{\Id_2}$.
	
	\item We partially measure each state $\ket{\omega }$  using  $\PP_U$, 
	obtaining the state $(\PP_U \otimes
	\II_2)(\proj{\omega}).$
	
	\item For each measurement $\PP_{U}$, we obtain one of two possible
	measurement results: ``$0$'' or ``$1$''. In consequence, we get $N_0$ 
	outcomes ``$0$''
	and $N_1$ outcomes ``$1$'', $N_0 + N_1 = N$. The state of the remaining 
	quantum
	part is equal to $\overline{P_{U,0}}^{\otimes N_0} \otimes
	\overline{P_{U,1}}^{\otimes N_1}$ (up to permutation of subsystems).
	Without loss of a generality (w.l.o.g.), we may assume that $N_0 \ge N_1$.
	
	\item By majority vote we reject minority report, \ie we reject all 
	outcomes ``$1$'' and 
	quantum states associated with them. As a result the memory state is given 
	by
	$\sigma_{\PP_U, \SS} = \overline{P_{U,0}}^{\otimes N_0}$.
	
	\item We prepare an arbitrary state $\rho \in \Omega(\HH_2)$.
	
	\item We perform a binary retrieval measurement $\mathcal{R} = \{ R,\Id - R 
	\}$ on $\rho \otimes
	\sigma_{\PP_U, \SS} $.
\end{enumerate}

To construct the effect $R$, let us fix $N_0$ and let $n= N_0-1$. We introduce 
the family of Dicke
states~\cite{mukherjee2020preparing}. The Dicke state $\ket{ D_k^n} $ is the 
$n$-qubit state, which
is equal to the superposition state of all ${ n \choose k }$ basis states of 
weight $k$. For
example, $\ket{ D_1^3} = \frac{1}{\sqrt{3}} \left( \ket{ 100} + \ket{ 010 } + 
\ket{001 } \right)$.
Let us also define
\begin{equation}
	s_n(k,m) \coloneqq \sum_{i=0}^k \sum_{j=0}^{n-k} \delta_{i+j-m} {k \choose
		i } { n-k \choose j } (-1)^{n-k-j }
\end{equation}
being the convolution of binomial coefficients. Consider the effect $R$ of the 
form
\begin{equation}\label{effect-r}
	R= \sum_{k=0}^{n} \proj{R_k},
\end{equation}
where $\ket{R_k} \coloneqq \frac{\ketv{M_k}}{||M_k||_2}$
and matrices $M_k \in \mathrm{M}\left(\HH_2, \HH_{2^{n+1}}\right)$ are given by
{\fontsize{9pt}{0}
	\begin{equation}
		M_k = \sum_{m = 0}^{n+1} \frac{s_n(k, n - m)\ket{0} + s_n(k, n + 1 -
			m)\ket{1}}{\sqrt{n + 1 \choose m}} \bra{ D_m^{n+1}}
\end{equation}}
for $k = 0, \ldots, n$. The proof that $R$ is a valid effect is relegated to
Lemma~\ref{measurement-correctly} in Appendix~\ref{app:pgls}. In this 
learning
scheme the approximation $\QQ_U = \{ Q_{U,0}, \Id_2 - Q_{U,0}\}$ is determined 
by relation $\tr
\left( \rho Q_{U,0} \right) = \tr \left( \left( \rho \otimes 
\overline{P_{U,0}}^{\otimes N_0}
\right) R \right)$. Basing on Lemma~\ref{form-of-q} in 
Appendix~\ref{app:pgls}, the effect
$Q_{U,0}$ has the form
\begin{equation}
	\begin{split}
		Q_{U,0} = \frac{N_0}{N_0+1} P_{U,0}.
	\end{split}
\end{equation}
Provided we observed $N_0$ outcomes ``$0$'', we have that $\mathcal{F}_2(\PP_U, 
\QQ_U) =
\frac{2N_0+1}{2N_0+2}$,  where $N_0$ satisfies $N_0 \ge 
\lceil{\frac{N}{2}}\rceil$. Note, that the
value of $\mathcal{F}_2(\PP_U, \QQ_U)$ does not depend on the choice of $U$. 
The 
average fidelity
function $\mathcal{F}_2^\text{avg}(\LL_{PGLS})$ defined for the pretty good 
learning scheme 
of qubit von Neumann
measurements satisfies 
\begin{equation}
	\begin{split}
		&\mathcal{F}_2^\text{avg}(\LL_{PGLS}) =\\
		&\begin{cases}
			\frac{1}{2^N} \sum\limits_{l=k}^{N} 2 {{N}\choose{l}} 
			\frac{2l+1}{2l+2} 
			, &N= 	2k-1,\\
			\frac{1}{2^N} \left( {{N}\choose{k}}  \frac{2k+1}{2k+2}  +  
			\sum\limits_{l=k+1}^{N} 2 {{N}\choose{l}} \frac{2l+1}{2l+2} \right) 
			, &N= 2k.
		\end{cases}
	\end{split}	 
\end{equation}
In the asymptotic regime, we may simplify the calculations to obtain
\begin{equation}
	\begin{split}
		\mathcal{F}_2^\text{avg}(\LL_{PGLS}) \ge
		\frac{2\lceil{\frac{N}{2}}\rceil +
			1}{2\lceil{\frac{N}{2}}\rceil + 2} = 1 - \Theta\left(\frac1N\right).
	\end{split}
\end{equation}

\begin{corollary}
	In the pretty good learning scheme $\LL_{PGLS} = \left( \sigma,\{ \CC_i 
	\}_{i=1}^{N-1}, \RR \right) $
	the initial state $\sigma$ is defined as a product of $N$ copies of 
	maximally entangled state
	$\ket{\omega}$, processing channels $\{ \CC_i \}_{i=1}^{N-1}$ are 
	responsible for majority
	voting and the measurement $\RR = \{ R,\Id - R \}$ is defined by 
	Eq.\eqref{effect-r}.
\end{corollary}

Finally, averaging the construction of $\QQ_U$ over all possible combinations 
of measurements'
results $\{ 0,1 \}^{N}$ leads to the following approximation of $\PP_U$.
\begin{corollary}
	The approximation $\QQ_U$ is a convex combination of the original 
	measurement $\PP_U$
	and the maximally depolarizing channel $ \Phi_*$. More precisely,
	\begin{equation}
		\QQ_U = p_0	\PP_{U} + (1-p_0) \Phi_*,
	\end{equation}
where $p_0 = 2\mathcal{F}_2^\text{avg}(\LL_{PGLS})-1$.
\end{corollary}

In the pretty good learning scheme, to keep the calculation simple, we 
assumed 
that the retrieval measurement 
$\RR$ uses a memory state that is a tensor product of $N_0$ copies of the 
same state ($\overline{P_{U,0}}$ or $\overline{P_{U,1}}$). The same 
approach was investigated in the paper~\cite{fiuravsek2002universal}, where
the value of the average fidelity function was originally derived
$\mathcal{F}_2(\PP_U, \QQ_U) = \frac{2N_0+1}{2N_0+2}$. However, one may 
improve the learning scheme by using all available output states 
$\overline{P_{U,0}}$ 
and $\overline{P_{U,1}}$. 
In that case, we expect to obtain  a higher 
value of the fidelity function. Such an intuition was confirmed by Gisin and 
Popescu~\cite{gisin1999spin}, who proved that a memory state 
$\overline{P_{U,0}} \otimes \overline{P_{U,1}}$ 
encodes more information of the effect $P_{U,0}$, than a state 
$\overline{P_{U,0}} \otimes \overline{P_{U,0}}$.

\subsection{Learning based on port-based teleportation}

We have observed in Section~\ref{fidelitylower} that learning scheme based on 
DPBT achieves the average fidelity $\mathcal{F}_2^{\text{avg}}(\LL_{DPBT}) = 1 
- \Theta\left(\frac{1}{N^2}\right).$ More precisely, from the proof presented 
in Appendix~\ref{app:lower-bound}, we get 
$\mathcal{F}_2^{\text{avg}}(\LL_{DPBT}) = \frac{1}{3} + \frac{2}{3}F_*$, where 
$F_*$ is the entanglement fidelity of DPBT protocol. For $d=2$ it is 
known~\cite{ishizaka2009quantum} that 
$F_* = \cos^2\left(\frac{\pi}{N+2}\right)$. Hence,
\begin{equation}
	\mathcal{F}_2^{\text{avg}}(\LL_{DPBT}) = \frac{1}{3} + 
	\frac{2}{3}\cos^2\left(\frac{\pi}{N+2}\right).
\end{equation}

The learning scheme can also be constructed  using probabilistic port-based 
teleportation (PPBT) (see for instance \cite{christandl2021asymptotic}). This 
protocol works similarly to $\LL_{DPBT}$ presented 
in Fig.~\ref{fig-dpbt-learning}. The difference is that the final measurement 
$\QQ_0$ returns a label $i \in \{0,\ldots, N \}$, where the label $ i > 0$ 
indicates the success of the teleportation procedure -- the initial state is in 
$i$-th port and the label $i=0$ indicates the protocol's failure. 
The result from~\cite{studzinski2017port} says that the
corresponding optimal probability of success teleportation is $p_0 = \frac{N}{N 
+ 
3}$. That means the approximation $\QQ_U$ achieved by the learning scheme 
$\LL_{PPBT}$is given by 
\begin{equation}
	\QQ_U =p_0 \PP_{U} + (1-p_0)\Phi_*,
\end{equation}
which implies
\begin{equation}
	\mathcal{F}_2^{\text{avg}}(\LL_{PPBT}) = \frac{2N+3}{2N+6}.
\end{equation}

\subsection{Numerical investigation}

In is generally difficult to find an optimal procedure for quantum operations learning. 
It is worth mentioning that the parallel learning schemes match adaptive ones 
for $N=1,2$ but for $N \ge 3$ adaptive strategies achieve slight 
advantage~\cite{bisio2011quantum}.

In the numerical analysis, we compared  average fidelity for the optimal 
parallel learning strategy $\LL_{\text{Parallel}}$ with the optimal adaptive 
strategy $\LL_{\text{Adaptive}}$. The scheme $\LL_{\text{Adaptive}}$ is also 
the best possible scheme available, which is causally structured that is 
$\mathcal{F}_2^{\text{avg}}(\LL_{\text{Adaptive}}) = F_2$. 

\begin{equation*}
\begin{array}{l|c|c|c|c|c}
	N&1&2&3&4&5\\\hline
	\mathcal{F}_2^{\text{avg}}(\LL_{\text{Adaptive}})&
	0.7499&
	0.8114&
	0.8684&
	0.8968&
	0.9189\\\hline
	\mathcal{F}_2^{\text{avg}}(\LL_{\text{Parallel}})&
	0.7499&
	0.8114&
	0.8676&
	0.8955&
	0.9187\\\hline
	\mathcal{F}_2^{\text{avg}}(\LL_{\text{DPBT}})&
	0.5000&
	0.6667&
	0.7697&
	0.8333&
	0.8745\\\hline
	\mathcal{F}_2^{\text{avg}}(\LL_{\text{PGLS}})&
	0.7500&
	0.7917&
	0.8438&
	0.8625&
	0.8854\\\hline
	\mathcal{F}_2^{\text{avg}}(\LL_{\text{PPBT}})&
	0.6250&
	0.7000&
	0.7500&
	0.7857&
	0.8125\\\hline
\end{array}
\end{equation*}

To optimize this problem we used the 
\texttt{Julia} 
programming language along with quantum package 
\texttt{QuantumInformation.jl}\cite{Gawron2018} and SDP optimization via SCS 
solver~\cite{ocpb:16, scs} with a precision $\epsilon = 10^{-8}$. The code is 
available on GitHub~\cite{code22}.

We compare the results obtained in this section in the Figure~\ref{fig:all}.

\begin{figure}[h!]
	\includegraphics[width=\linewidth]{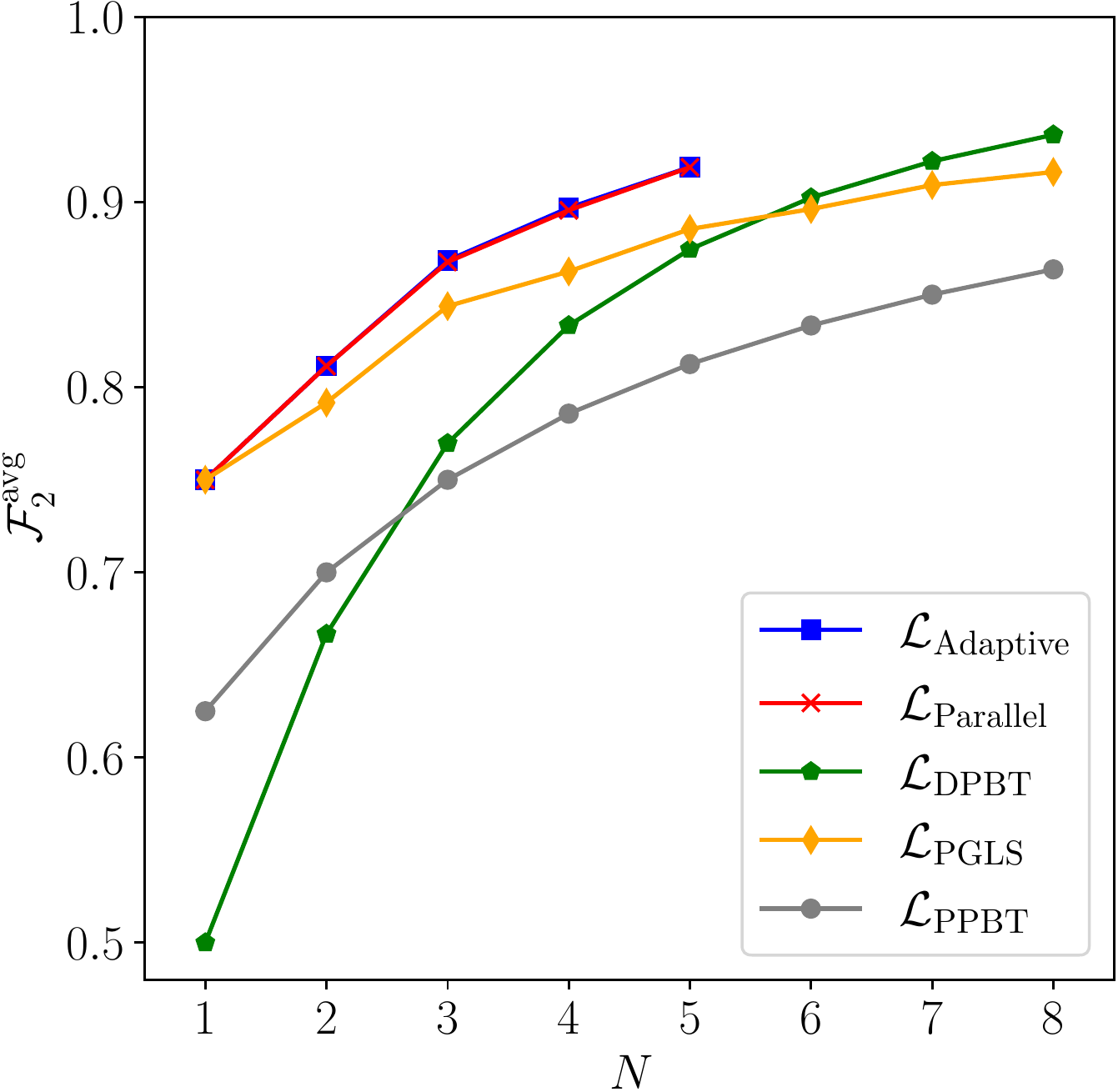}
	\caption{The average fidelity function (dimensionless) 
	$\mathcal{F}_2^{\text{avg}}$ 
	calculated for $N \rightarrow 1$ learning scheme: optimal 
	adaptive strategy, $\LL_{\text{Adaptive}}$ (numerical value, blue squares); 
	optimal parallel scheme, $\LL_{\text{Parallel}}$ (numerical value, red 
	crosses); learning scheme based on DPBT, $\LL_{\text{DPBT}}$ 
	(green pentagons); pretty good learning scheme, 
	$\LL_{\text{PGLS}}$ 
	(orange diamonds); learning scheme based on PPBT, $\LL_{\text{PPBT}}$ 
	(gray circles).
	}\label{fig:all}
\end{figure}

\section{Conclusions and discussion}\label{sec:conclusion}
In this work, we studied  the problem of learning an unknown von 
Neumann measurement of dimension $d $ from a finite number of copies.
Our goal was to find the asymptotic behavior of the maximum value for the 
average fidelity 
function $F_d$. This value was maximized over all possible learning schemes, 
and the average was taken over all von Neumann
measurements.  By using the deterministic PBT protocol, we were able 
to state the lower bound $1 - \Theta\left(\frac{1}{N^2}\right)$, which matched 
the obtained upper bound and hence, solved the given problem.

In the qubit case, we introduced a scheme called the pretty good learning 
scheme. This scheme was a particular case of a parallel learning protocol, and 
it used only two-qubit entangled memory states. The average fidelity function 
calculated for the pretty good learning scheme behaved as 
$1-\Theta\left(\frac{1}{N}\right)$. Moreover, we compared the performance of 
different learning schemes: adaptive, parallel, based on DPBT, based on PPBT 
and the pretty good 
learning scheme for the qubit case. Although, the learning scheme based on PBT 
were asymptotically optimal, we showed that the pretty good learning scheme 
outperforms it for low values of $N$.

This work paves the way toward a complete description of the capabilities of von 
Neumann measurement learning schemes. One potential way forward is the 
probabilistic storage and 
retrieval approach,
widely studied for unitary operations and phase rotations in 
\cite{sedlak2019optimal,
	sedlak2020probabilistic}. According to our numerical results, the 
	probability of retrieval of a
quantum measurement in a parallel scheme is exactly $N/(N+3)$, which 
corresponds to the value
obtained in \cite{sedlak2019optimal} for unitary channels, while adaptive 
strategies for quantum
measurements learning to provide slightly higher probability, starting from $N \ge 
3$.

\section*{Acknowledgments}

We would like to thank the  anonymous reviewer for insightful comments and 
suggestions, especially for introducing us the concept of port-based 
teleportation, which provided the asymptotically optimal lower bound for the 
fidelity value. This work was supported by the project „Near-term quantum 
computers Challenges, 
optimal implementations and applications”
Grant No. POIR.04.04.00-00-17C1/18-00, which is carried out within the 
Team-Net programme of the
Foundation for Polish Science co-financed by the European Union under the 
European Regional
Development Fund.
Paulina Lewandowska and Ryszard Kukulski are holders of European Union scholarship through the European Social Fund, grant InterPOWER (POWR.03.05.00-00-Z305).

\bibliographystyle{ieeetr}
\bibliography{learning.bib}

\onecolumngrid
\newpage
\appendix

\section{Proof of lower bound}\label{app:lower-bound}
\begin{lemma}\label{lemma-dpbt}
	Let us fix $d \in \N$ and let $\LL$ be a parallel learning scheme based on 
	the DPBT protocol introduced in Section~\ref{fidelitylower}. It holds that
	\begin{equation}
	\mathcal{F}_d^{\text{avg}}(\LL) = 1 - 
		\Theta\left(\frac{1}{N^2}\right).
	\end{equation}
\end{lemma}
\begin{proof}
	Let $\proj{\psi_0}$ and $\QQ_0 = \{Q_{0,i}\}_{i=1}^N$ be a 
	realization of the optimal teleportation strategy $\TT_0$, such 
	that~\cite{christandl2021asymptotic}
	\begin{equation}
		F_* = \frac{1}{d^2}\tr \left(T_0 \projv{\Id_d}\right) = 1 - 
		\Theta\left(\frac{1}{N^2}\right).
	\end{equation}
	Let us introduce the operations $\QQ_{0,i}(\sigma) = \sqrt{Q_{0,i}} \sigma 
	\sqrt{Q_{0,i}}$. 
	Then, the 
	approximation $\QQ_U$ acting on an arbitrary state $\rho$ can be 
	expressed as
	\begin{equation}
	\begin{split}
		\QQ_U(\rho) &= \sum_{i=1}^N \tr_{\HH_d^{(in)} \otimes  \HH_{d^N}^{(A)} 
		\otimes 
			\HH_{d^{N-1}}^{(\bar B_i)}} \left( \left( \QQ_{0,i} \otimes 
			\PP_U^{\otimes N} 
		\right)(\rho \otimes \proj{\psi_0})\right)\\
		 &= \PP_{U} \left(\sum_{i=1}^N \tr_{\HH_d^{(in)} \otimes  
		 \HH_{d^N}^{(A)} 
			\otimes 
			\HH_{d^{N-1}}^{(\bar B_i)}} \left( \left( \QQ_{0,i} \otimes 
		\II_{d^N} 
		\right)(\rho \otimes \proj{\psi_0})\right)\right) = \PP_U(\TT_0(\rho)).
	\end{split}
	\end{equation}
	Let $J_\Delta = \sum_{ i=0}^{d-1} \proj{i} \otimes 
	\proj{i}$ be the Choi matrix of the completely dephasing channel $\Delta$. By using the equality $\mathcal{F}_d(\PP_U, \QQ_U) = \frac{1}{d} 
	\tr(P_U 
	Q_U)$ we obtain
	\begin{equation}
	\begin{split}
		\mathcal{F}_d^{\text{avg}}(\LL) &= \frac{1}{d} \int_U dU  \tr(P_U 
		Q_U) = \frac{1}{d} \int_U dU  \tr\left((U \otimes \overline{U}) 
		J_\Delta (U \otimes \overline{U})^\dagger T_0\right)\\
		&= \frac{1}{d}  \tr\left( \frac{1}{d+1} (\Id_{d^2} + \projv{\Id_d})   
		T_0\right) = \frac{1}{d+1} + \frac{d}{d+1}F_*\\
		&=1 - \Theta\left(\frac{1}{N^2}\right), 
	\end{split}
	\end{equation}
	which completes the proof.
\end{proof}

\section{Proof of Lemma~\ref{lem-upper-main} }\label{app:upper-bound-2}
\begin{figure}[h!]
	\centering\includegraphics{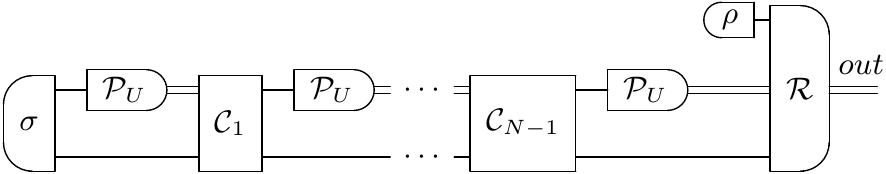}
	\caption{ The schematic representation of a learning scheme $\LL =
		\left( \sigma, \{ \CC_i
		\}_{i=1}^{N-1}, \RR \right)$.
		\label{fig-app-gen}}
\end{figure}

Let $d=2$ and let us fix $N \in \N$. In the $N \rightarrow 1$ learning scheme of 
single-qubit 
von Neumann 
measurements we have access to $N$ copies of a given measurement $\PP_U$, 
which 
is parameterized by 
some unitary matrix $U \in \mathrm{M}(\HH_2)$. Let us consider a general 
single-qubit von Neumann 
measurement learning scheme $\LL$, 
which is depicted in Fig.~\ref{fig-app-gen}. The Choi-Jamio{\l}kowski 
representation of $\LL$ is 
given as $L = \sum_{ i=0}^1 \proj{i} \otimes L_i$, where $\ket{i} \in 
\HH_{2}^{(out)}$. The result of composition of all copies of $\PP_U$ and
the scheme $\LL$ is a measurement $\QQ_U = \{Q_{U, 0}, Q_{U,1}\}$, which is an 
approximation of
$\PP_U$. To define the effects $Q_{U,i}$ we use the link product 
\cite{bisio2016quantum} in the following way 
$\tr(\rho Q_{U,i}) = \tr\left(L_i^\top \left(\rho \otimes P_U^{\otimes 
N}\right)\right)$ for $\rho
\in \Omega(\HH_2)$ and $i=0,1$. Thus, we can calculate the fidelity defined in
Eq.~\eqref{fidelity} between $\PP_U$ and $\QQ_U$
\begin{equation}
	\mathcal{F}_2(\PP_U, \QQ_U) = \frac{1}{2} \sum_{i=0}^1
	\tr(P_{U, i} Q_{U, i}) = \frac{1}{2} \sum_{i=0}^1
	\tr\left[L_i^\top \left(P_{U, i} \otimes P_U^{\otimes N}\right)\right].
\end{equation}
Finally, we can express the maximum value of the average fidelity function 
$F_2$ 
defined 
in Eq.~\eqref{eq:fidelity} as
\begin{equation}\label{eq:fidelity-appendix}
	F_2 = \max_\LL \int_U dU \frac{1}{2} \sum_{i=0}^1
	\tr\left[L_i^\top \left(P_{U, i} \otimes P_U^{\otimes N}\right)\right].
\end{equation}
In the following subsections we will upper bound $F_2$ by using this simplified 
maximization formula.

\subsection{Measurement learning via parallel storage of unitary
	transformations}\label{upper:simplification-of-problem}

\begin{figure}[h!]
	\centering\includegraphics{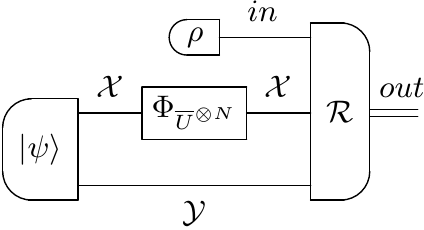}
	\caption{Schematic representation of the setup, which we use to upper bound 
	$F_2$. In this 
		scenario, we are given $N$ copies of unitary channel $\Phi_{\bar U}$ in 
		parallel. Our objective 
		is to approximate the von Neumann measurement $\PP_U$. 
		\label{upper-bound-proof-3} }
\end{figure}

In this section we consider a new learning scheme, presented in 
Fig.~\ref{upper-bound-proof-3}. In 
this scheme, we are given $N$ copies of a unitary channel, $\Phi_{\bar U}$, 
which 
we can use in 
parallel. Our goal is to approximate the measurement $\PP_U$ using the black 
box with the unitary channel $\Phi_{\bar U}$ inside. To achieve this, we 
choose an 
initial memory 
state $\ket{\psi} \in \XX \otimes \YY$ and a retrieval binary measurement $\RR 
= \{R_0, R_1\}$, such 
that $R_i \in \mathrm{M}(\ZZ \otimes \XX \otimes \YY)$, where $\ZZ = 
\HH_2^{(in)}, \XX = \HH_{2^N}$ 
and $\YY = \HH_{2^N}$. We maximize the value of the average fidelity 
function, which will be denoted as $\widetilde{F_2}$. To calculate $\widetilde{F_2}$ 
we may 
observe that for 
a 
given $\rho \in 
\Omega(\ZZ)$, the probability that outcome $i$ occurs is equal $\tr \left( R_i 
\left( 
\rho \otimes ({\bar 
	U}^{\otimes N} \otimes \Id_a) \proj{\psi}({U^\top}^{\otimes N} \otimes 
	\Id_a)\right) \right)$, 
where  $a \coloneqq 2^N$. Therefore, we obtain
\begin{equation}\label{eq:fidelity-new-appendix}
	\widetilde{F_2} = \max_{\substack{\RR= \{R_0, R_1\}\\
			\proj{\psi} \in \Omega(\XX \otimes \YY)}} \int_U dU \frac{1}{2}
	\sum_{i=0}^1
	\tr \left[ R_i \left((U \otimes {\bar U}^{\otimes N} \otimes \Id_a)(
	\proj{i} \otimes  \proj{\psi})(U^\dagger \otimes {U^\top}^{\otimes
		N} \otimes \Id_a)\right) \right].
\end{equation}

\begin{lemma}\label{lemma-unitary-transformation}
	Let $F_2$ be the fidelity function defined in 
	Eq.~\eqref{eq:fidelity-appendix}
	and $\widetilde{F_2}$ be the fidelity function defined in
	Eq.~\eqref{eq:fidelity-new-appendix}. Then, it holds that $F_2 \le 
	\widetilde{F_2}$.
\end{lemma}

\begin{proof}
	First, we observe that each von Neumann measurement $\PP_U$ can be written as 
	a composition of the
	completely dephasing channel $\Delta$ given by $\Delta(X) = \sum_{ i=0}^1 
	\bra{i}X\ket{i} \proj{i}$,
	and a unitary channel $\Phi_{U^\dagger}$. Equivalently, that means $P_U = 
	(\Delta \otimes
	\II_2)\left(\projv{U^\dagger}\right)$. Due to the fact that the channel 
	$\Delta$ is 
	self-adjoint, we obtain
	\begin{equation}
		\tr\left[L_i^\top \left(P_{U, i} \otimes P_U^{\otimes N}\right)\right] =
		\tr\left[\left(( \II_2 \otimes (\Delta \otimes  \II_2)^{\otimes
			N})(L_i)\right)^\top
		\left(P_{U, i} \otimes \projv{U^\dagger}^{\otimes N}\right)\right].
	\end{equation}
	Note that $\sum_{i=0}^1 \proj{i} \otimes ( \II_2 \otimes (\Delta \otimes 
	\II_2)^{\otimes N})(L_i)$
	represents the composition of the learning scheme $\LL$ and $N$ copies of channels $\Delta$. 
	If we omit processing
	channels $\Delta$, we get the following upper bound on $F_2$ defined in
	Eq.~\eqref{eq:fidelity-appendix}
	\begin{equation}\label{fidelity-upper-proof}
		\begin{split}
			F_2 &\le \max_\LL \int_U dU \frac{1}{2} \sum_{i=0}^1
			\tr\left[L_i^\top \left(P_{U, i} \otimes \projv{U^\dagger}^{\otimes
				N}\right)\right] \\&=
			\frac12 \max_\LL \int_U dU \tr\left[L^\top \left((\Id_2 \otimes
			U) J_\Delta (\Id_2 \otimes U^\dagger) \otimes
			\projv{U^\dagger}^{\otimes N}\right)\right],
		\end{split}
	\end{equation}
	where $J_\Delta$ is Choi-Jamio{\l}kowski representation of $\Delta$. 
	Observe that the maximal value
	of the integral in above equation is achievable by networks $\LL$ which 
	satisfy the following
	commutation relation
	\begin{equation}\label{eq-comm-proof}
		[L,\Id_2 \otimes \bar U \otimes (\Id_2 \otimes U)^{\otimes N}] = 0,
	\end{equation}
	for any unitary matrix $U$. To argue this fact, for any $\LL$ one can
	define a learning network $\tilde \LL$ given by
	\begin{equation}
		\tilde L = \int_U dU \left((\Id_2 \otimes \bar U) \otimes (\Id_2 \otimes
		U)^{\otimes N}\right) L \left((\Id_2 \otimes U^\top) \otimes (\Id_2 
		\otimes
		U^\dagger)^{\otimes N}\right).
	\end{equation}
	It is not difficult to see that $\tilde L$ is a correctly defined 
	Choi-Jamio{\l}kowski
	representation of a quantum learning network~\cite[Theorem 
	2.5]{bisio2016quantum}, which satisfies
	the relation Eq.~\eqref{eq-comm-proof}. Moreover, for both $L$ and $\tilde 
	L$ the value of the
	integral in Eq.~\eqref{fidelity-upper-proof} remains the same.
	
	\begin{figure}[h!]
		\centering\includegraphics{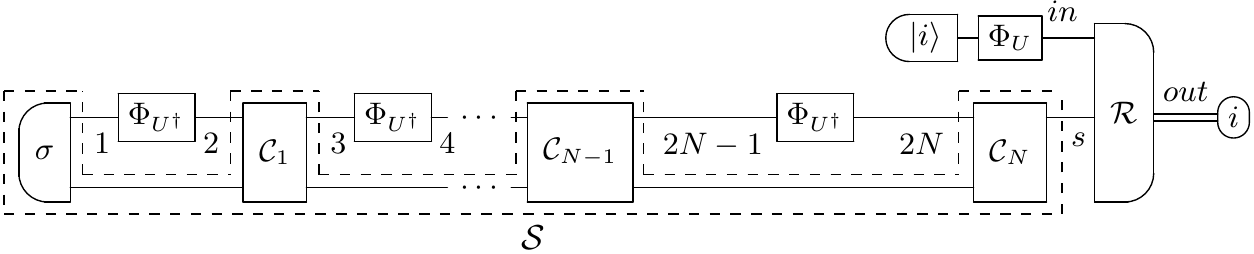}
		\caption{Schematic representations of the right-hand side of
			Eq.~\eqref{fidelity-upper-proof}. With probability $1/2$ we prepare 
			one of basis states
			$\ket{0}$ or $\ket{1}$ and calculate the probability that we obtain 
			output
			$i$. Eq.~\eqref{fidelity-upper-proof} is then the cumulative 
			probability that
			provided the state $\proj{i}$ we measure $i$. The learning scheme 
			$\LL$ is
			given as $\LL = \left( \sigma, \{ \CC_i \}_{i=1}^{N}, \RR \right)$ 
			and
			the storage $\SS$ (marked with a dashed line) is defined as a 
			composition of an initial memory 
			state $\sigma$ and processing channels  $\{ \CC_i \}_{i=1}^{N}$.
			\label{upper-bound-proof-2}}
	\end{figure}
	
	Let us divide $\LL$ into a storage network $\SS$ and a retrieval 
	measurement $\mathcal{R}$, as shown
	in Fig.~\ref{upper-bound-proof-2}. We introduce the input space $\XX_{I} 
	\coloneqq
	\bigotimes_{i=1}^{N} \HH_{2}^{(2k)}$ (denoted with numbers $2,4,\ldots,2N$ 
	on
	Fig.~\ref{upper-bound-proof-2}) and the output space $\XX_{O} \coloneqq 
	\bigotimes_{i=1}^{N}
	\HH_{2}^{(2k-1)}$ (denoted with numbers $1,3,\ldots,2N-1$). Additionally, 
	we define spaces
	$\HH_2^{(in)}, \HH_{2}^{(out)}$ and $\HH_s$. The space $\HH_s$ has 
	arbitrary dimension $s$, but not
	smaller than the dimension of $\XX_I \otimes \XX_O$. The storage $\SS$ can 
	be realized as a sequence
	of isometry channels followed by a partial trace operation \cite[Theorem 
	2.6]{bisio2016quantum}.
	Therefore, by moving the partial trace operation to the retrieval part, 
	$\RR$, we may assume that
	the storage $\SS$ consists of an initial pure state followed by a sequence 
	of isometry channels.
	In consequence, the Choi-Jamio{\l}kowski matrix of $\SS$ has the form $S = 
	\projv{X}$. There exists 
	an isometry $V \in \mathrm{M}\left(\HH_s, \XX_{I} \otimes \XX_{O}\right)$, 
	such that $X
	=\sqrt{\tr_{\HH_s}S}V^\top$. In this notation, $S$ is the solution of $S = 
	(\Id_{4^N} \otimes
	V) \projv{\sqrt{\tr_{\HH_s}S}} (\Id_{4^N} \otimes V)^\dagger.$ Hence, the 
	isometry channel $V \cdot
	V^\dagger$ can be treated as a postprocessing of the storage $\SS$ and also 
	viewed as a part of the
	retrieval $\RR$. In summary, after all changes, the storage $\SS$ is of the 
	form $S =
	\projv{\sqrt{\tr_{\HH_s}S}}$. By using the normalization property 
	\cite[Theorem
	2.5]{bisio2016quantum} for the network presented in 
	Fig.~\ref{upper-bound-proof-2}, we obtain
	$\tr_{\HH_{2}^{(out)}} L = \Id_{2} \otimes \tr_{\HH_s}S.$ Therefore, using 
	the property
	Eq.~\eqref{eq-comm-proof} we have
	\begin{equation}\label{eq-comm-proof-2}
		[\tr_{\HH_s}S, (\Id_2 \otimes U)^{\otimes N}] = 0.
	\end{equation}
	Let us define the memory state $\sigma_{\Phi_{U^\dagger}, \SS}$ as an
	application of the
	storage $\SS$ on $N$ copies of $\Phi_{U^\dagger}$. Then, we have
	\begin{equation}
		\begin{split}
			\sigma_{\Phi_{U^\dagger}, \SS} &= \tr_{\XX_I \otimes \XX_O}
			\left[\projv{\sqrt{\tr_{\HH_s}S}}
			\left(\projv{U^\top}^{\otimes N} \otimes \Id_{4^N}\right)\right] 
			\\&=
			\tr_{\XX_I
				\otimes \XX_O} \left[\projv{(\Id_2 \otimes
				U^\dagger)^{\otimes N}\sqrt{\tr_{\HH_s}S}}
			\left(\projv{\Id_2}^{\otimes N} \otimes \Id_{4^N}\right)\right]
			\\	&= \left(\Id_2 \otimes \bar U \right)^{\otimes N}\proj{\psi}
			\left(\Id_2 \otimes
			U^\top\right)^{\otimes N},
		\end{split}
	\end{equation}
	where  the last equality uses  the property Eq.~\eqref{eq-comm-proof-2} 
	and introduced
	$\ket{\psi} \coloneqq \left(\brav{\Id_2}^{\otimes N} \otimes \Id_{4^N} 
	\right)
	\ketv{\sqrt{\tr_{\HH_s}S}}$. It means that an arbitrary storage 
	strategy $\SS$, which has
	access to $N$ copies of a unitary channel $\Phi_{U^\dagger}$ can be 
	replaced with parallel storage
	strategy of $N$ copies of a unitary channel $\Phi_{\bar U}$. By exploiting 
	this property to
	Eq.~\eqref{fidelity-upper-proof} we obtain
	\begin{equation}
		\begin{split}
			F_2 &\le 	\frac12 \max_\LL \int_U dU \tr\left[L^\top \left((\Id_2 
			\otimes
			U) J_\Delta (\Id_2 \otimes U^\dagger) \otimes
			\projv{U^\dagger}^{\otimes N}\right)\right]
			\\&= \frac12 \max_{\substack{\RR= \{R_0, R_1\}\\
					\SS}} \int_U dU \sum_{i=0}^1 \tr\left[R_i 
					(U\proj{i}U^\dagger \otimes
			\sigma_{\Phi_{U^\dagger}, \SS})\right]
			\\&= \frac12 \max_{\substack{\RR= \{R_0, R_1\}\\
					\proj{\psi} \in \Omega(\XX_I \otimes \XX_O)}} \int_U dU 
					\sum_{i=0}^1
			\tr\left[R_i 	\left(U\proj{i}U^\dagger \otimes (\Id_2 \otimes \bar
			U)^{\otimes N}\proj{\psi} (\Id_2 \otimes
			U^\top)^{\otimes N} \right)\right] = \widetilde{F_2}.
		\end{split}
	\end{equation}
\end{proof}

\subsection{Objective function simplification}\label{upper:standarization}

The aim of this section is to simplify the maximization of the fidelity 
function $\widetilde{F_2}$ defined in
Eq.~\eqref{eq:fidelity-new-appendix}. Let us consider a binary measurement 
$\RR 
= \{R_0, R_1\}$
taken from the maximization domain in Eq.~\eqref{eq:fidelity-new-appendix}. It 
holds that $R_0 + R_1
= \Id_{2^{2N+1}}$, and hence we may write
\begin{equation}\label{fidelity-upper-proof-2}
	\begin{split}
		\widetilde{F_2} &= \max_{\substack{\RR= \{R_0, R_1\}\\
				\proj{\psi} \in \Omega(\XX \otimes \YY)}} \int_U dU \frac{1}{2}
		\sum_{i=0}^1
		\tr \left[ R_i \left((U \otimes {\bar U}^{\otimes N} \otimes \Id_a)(
		\proj{i} \otimes  \proj{\psi})(U^\dagger \otimes {U^\top}^{\otimes
			N} \otimes \Id_a)\right) \right] \\
		&= \frac12 + \frac12 \max_{\substack{\RR= \{R_0, R_1\}\\
				\proj{\psi} \in \Omega(\XX\otimes\YY)}} \int_U dU \tr \left[ R_0
		\left((U \otimes
		{\bar U}^{\otimes N} \otimes \Id_a)(
		\sigma_z \otimes  \proj{\psi})(U^\dagger \otimes {U^\top}^{\otimes
			N} \otimes \Id_a)\right) \right]\\
		&= \frac12 + \frac12 \max_{\substack{\RR= \{R_0, R_1\}\\
				\proj{\psi} \in \Omega(\XX \otimes \YY)}} \tr \left[ \int_U dU
		(U^\dagger
		\otimes
		{U^\top}^{\otimes
			N} \otimes \Id_a) R_0(U \otimes \bar U^{\otimes N}	\otimes \Id_a) 
			(\sigma_z
		\otimes \proj{\psi})\right],
	\end{split}
\end{equation}
where $\sigma_z = \proj{0} - \proj{1}$. Observe that, taking the average value of the matrix $R_0$ 
over the unitary
group $\{ U \otimes \bar U^{\otimes N}	\otimes \Id_a \}_U$ is equivalent to 
taking $R$ such that $
0 \le R \le \Id_{2^{2N+1}}$ and  $[R, U \otimes \bar U^{\otimes N} \otimes 
\Id_a] = 0$ for any qubit
unitary matrix $U$. Equivalently, we may write $[R^{\top_\ZZ}, U^{\otimes N+1} 
\otimes \Id_a] = 0$,
where $\cdot^{\top_\ZZ}$ represents the partial transposition over subsystem 
$\ZZ$. According to
\cite[Theorem 7.15]{watrous2018theory} the matrix $R^{\top_\ZZ}$ commutes with 
$ U^{\otimes N+1}
\otimes \Id_a$ if and only if it is of the form
\begin{equation}
	R^{\top_\ZZ} = \sum_{\pi} W_\pi \otimes M_\pi,
\end{equation}
where matrices $W_\pi \in \mathrm{M}(\ZZ \otimes \XX)$ represent subsystem 
permutation matrices
acting on $N+1$ qubit systems, according to the equation
\begin{equation}
	W_\pi\ket{b_0, b_1,\ldots, b_N} = \ket{b_{\pi(0)}, b_{\pi(1)}, \ldots,
		b_{\pi(N)}}, \, b_k \in \{0,1\}.
\end{equation}
The matrices $M_\pi$ belong to the set $\mathrm{M}(\YY)$ and the index $\pi$ 
goes over all
permutations of the set $\{0,\ldots,N\}$. Hence, we may simplify calculation 
of 
$\widetilde{F_2}$
\begin{equation}\label{fidelity-upper-proof-3}
	\widetilde{F_2} = \frac12 + \frac12 \max_{\substack{R: \,\, 0 \le R \le 
	\Id_{2^{2N+1}}\\
			R = \sum_{\pi} W_\pi^{\top_\ZZ} \otimes M_\pi\\
			\proj{\psi} \in \Omega(\XX \otimes \YY)}} \tr\left[R(\sigma_z
	\otimes \proj{\psi}) \right].
\end{equation}

To simplify the calculation of $\widetilde{F_2}$ even further, we introduce the 
following 
notation of basis
states defined on $N + 1$ qubit system with fixed weight. We enumerate qubit 
subsystems with numbers
$0,1,\ldots,N$. For any subset $A_k \subset \{1,\ldots,N\}$, such that $|A_k| 
= 
k$ we define:
\begin{equation}
	\HH_{2^N} \ni \ket{A_k} \coloneqq \bigotimes_{i = 1} ^ N (\delta(i \in A_k)
	\ket{1} +
	\delta(i
	\not\in A_k) \ket{0}).
\end{equation}
Consider the following subspaces of the $N+1$ qubit space:
\begin{equation}
	\HH^{(k)} \coloneqq \text{span}\left(\ket{0} \ket{A_k}, \ket{1}
	\ket{A_{k + 1}}: A_k, A_{k+1} \subset \{1,\ldots,N\}\right)
\end{equation}
for $k = -1,\ldots,N$, where the vectors exist if and only if the expression is 
well-defined (for instance,
the vectors $\ket{A_{-1}}, \ket{A_{N + 1}}$ do not exist). In this notation, 
subspaces $\HH^{(k)}$
constitute a decomposition of $N+1$ qubit space, $\HH_{2^{N+1}} = \bigoplus_{k 
= -1}^N \HH^{(k)}$.
One may observe, that the matrix $R$ appearing  in the maximization domain 
of
Eq.~\eqref{fidelity-upper-proof-3} is block diagonal in the introduced 
decomposition (in the
partition $\ZZ \otimes \XX / \YY$). For such retrieval $R$, let us consider
\begin{equation}\label{fidelity-upper-proof-4}
	H_R = \tr_{\ZZ} \left(R(\sigma_z \otimes \Id_{4^N}) \right).
\end{equation}
Observe that the matrix $H_R$ is block diagonal in the decomposition
\begin{equation}
	\HH_{2^N} = \bigoplus_{k = 0}^N \text{span}(\ket{A_k}: A_k \subset
	\{1,\ldots,N\}).
\end{equation}
Hence, we will write $H_R$ as
\begin{equation}\label{fidelity-upper-proof-5}
	H_R = \bigoplus_{k = 0}^N H_{R,k}.
\end{equation}
Utilizing the above observations, the maximization problem 
Eq.~\eqref{fidelity-upper-proof-3} can be
written as
\begin{equation}\label{fidelity-upper-proof-6}
	\begin{split}
		\widetilde{F_2} &= \frac12 + \frac12 \max_{\substack{R: \,\, 0 \le R \le
				\Id\\ R = \sum_{\pi} W_\pi^{\top_\ZZ} \otimes M_\pi\\
				\proj{\psi} \in \Omega(\XX \otimes \YY)}} \tr\left[R(\sigma_z
		\otimes \proj{\psi}) \right]=\frac12 + \frac12 \max_{\substack{R: \,\, 0
				\le R \le \Id\\ R = \sum_{\pi} W_\pi^{\top_\ZZ} \otimes M_\pi\\
				\proj{\psi} \in \Omega(\XX \otimes \YY)}} \bra{\psi} H_R 
				\ket{\psi}\\
		&=\frac12 + \frac12 \max_{k = 0,\ldots,N} \max_{\substack{R: \,\, 0 \le 
		R
				\le \Id\\ R = \sum_{\pi} W_\pi^{\top_\ZZ} \otimes M_\pi
		}}  \lambda_1(H_{R,k})
	\end{split}
\end{equation}
where $\lambda_1(\cdot)$ stands for the largest eigenvalue and we used 
shortcut 
$\Id =
\Id_{2^{2N+1}}$. Finally, we observe that $H_R = -(\sigma_x^{\otimes N} 
\otimes 
\Id_a) H_R
(\sigma_x^{\otimes N} \otimes \Id_a)$, where $\sigma_x = \ketbra{0}{1} + 
\ketbra{1}{0}$. It implies
that $H_{R,k}$ is unitarily equivalent to $-H_{R,N-k}$ for any $k$. We use 
this 
fact to write the
final simplification of $\widetilde{F_2}$. The following lemma sums up all the 
considerations presented in this
section.

\begin{lemma}\label{lemma-simply}
	For the fidelity function $\widetilde{F_2}$ defined in 
	Eq.~\eqref{eq:fidelity-new-appendix}
	it holds that
	\begin{equation}\label{fidelity-upper-proof-7}
		\widetilde{F_2} = \frac12 + \frac12 \max_{k = 0,\ldots, \floor{N/ 2}}
		\max_{\substack{R: \,\, 0 \le R \le \Id\\ R = \sum_{\pi} 
		W_\pi^{\top_\ZZ}
				\otimes M_\pi }}  \|H_{R,k}\|_\infty.
	\end{equation}
\end{lemma}

\subsection{Technical lemmas}\label{ss:2k}
In the following lemma we will observe that optimization problem in  
Eq.~\eqref{fidelity-upper-proof-7} can be reduced to the case $k \in \N, N = 
2k$.
\begin{lemma}\label{lemma-last}
	Let $N \in \N$ and take $k$, such that $ k \le N/2$. It holds that
	\begin{equation}\label{lem-ap-N}
		\max_{\substack{R: \,\, 0 \le R \le \Id\\
				R = \sum_{\pi} W_\pi^{\top_\ZZ} \otimes M_\pi
		}} \|H_{R,k}\|_\infty \le \max_{\substack{\widetilde R: \,\, 0 \le
				\widetilde R \le
				\Id\\
				\widetilde R = \sum_{\pi} \widetilde W_\pi^{\top_\ZZ} \otimes 
				\widetilde M_\pi
		}} \|\widetilde H_{\widetilde R,N-k}\|_\infty,
	\end{equation}
	where the matrix $\widetilde R$ is defined for $\widetilde N = 2(N-k)$ and 
	hence the number of
	systems on which the matrix $\widetilde W_\pi$ acts is $\widetilde N + 1$.
\end{lemma}
\begin{proof}
	Let us fix $R$ such that $0\le R \le \Id$ and $R = \sum_{\pi} 
	W_\pi^{\top_\ZZ} \otimes M_\pi$. 
	Define 
	\begin{equation}
		\widetilde R \coloneqq \sum_{\pi} \left(W_\pi^{\top_\ZZ} \otimes 
		\Id_{2^{N - 2k}}\right) 
		\otimes \left(M_\pi \otimes \Id_{2^{N - 2k}}\right).
	\end{equation} 
	We see that matrix $\widetilde{R}$ is in the maximization domain of the 
	right-hand side of
	Eq.~\eqref{lem-ap-N}. Then, we have $\widetilde H_{\widetilde{R}} = \tr_\ZZ 
	\left(\widetilde
	R(\sigma_z \otimes \Id) \right) = \bigoplus_l \widetilde 
	H_{\widetilde{R},l}$. The matrix $
	\widetilde H_{\widetilde{R},N-k}$ is defined on the space spanned by the 
	vectors $\ket{A_{N - k}} 
	\in
	\HH_{2^{\widetilde N}}$ for $A_{N - k} \subset \{1,\ldots,\widetilde N\}$. 
	These vectors can be
	expressed in the form $	\ket{A_{N - k}}= \ket{B_i} \ket{B_{N-k-i}},$ where 
	$\ket{B_i} \in
	\HH_{2^{N}} $ for $B_i$ such that $|B_i|=i$, $B_i \subset \{1,\ldots,N\}$, 
	and 
	$\ket{B_{N-k-i}}\in
	\HH_{2^{N-2k}}$, $B_{N-k-i} \subset \{N+1,\ldots,\widetilde N\}$. Then, we 
	have
	\begin{equation}
		\begin{split}
			(\bra{A_{N-k}} \otimes \Id) \widetilde H_{\widetilde R, N-k}
			(\ket{A'_{N-k}} \otimes
			\Id) = \braket{B_{N-k-i}}{B'_{N-k-i'}}\left(\bra{B_{i}}\otimes 
			\Id\right)
			H_R\left(\ket{B'_{i'}}\otimes \Id\right) \otimes \Id.
		\end{split}
	\end{equation}
	The non-zero blocks exist if and only if $i = i'$ and $B_{N-k-i} = 
	B'_{N-k-i'}$,
	so
	\begin{equation}
		\widetilde H_{\widetilde R, N-k} = \bigoplus_{i = k}^{N-k}
		\bigoplus_{\substack{B_{N-k-i}:\\ B_{N-k-i} \subset 
		\{N+1,\ldots,\widetilde
				N\} }} H_{R,i} \otimes \Id.
	\end{equation}
	That means
	\begin{equation}
		\|\widetilde H_{\widetilde R, N-k}\|_\infty = \max_{i = k,\ldots,N-k} \|
		H_{R,i}\|_\infty \ge \|H_{R,k}\|_\infty.
	\end{equation}
\end{proof}

In the next lemma we will find the upper bound for 
Eq.~\eqref{fidelity-upper-proof-7} in the case $N
= 2k$ for $k \in \N$.

\begin{lemma}\label{lemma-k-2k}
	Let $k \in \N$ and $N = 2k$. For matrices $R$ and $H_{R,k}$ defined in
	Subsection \ref{upper:standarization} we have
	\begin{equation}
		\max_{\substack{R: \,\, 0 \le R \le \Id\\ R = \sum_{\pi} 
		W_\pi^{\top_\ZZ}
				\otimes M_\pi }}  \|H_{R,k}\|_\infty \le 1 -
		\Theta\left(\frac{1}{k^2}\right).
	\end{equation}
\end{lemma}
\begin{proof}
	Let us fix $R$ such that $0\le R \le \Id$ and $R = \sum_{\pi} 
	W_\pi^{\top_\ZZ} \otimes M_\pi$. 
	Through
	the rest of the proof, by $B_l$ we denote subsets of $\{1,\ldots,2k\}$, 
	such that $|B_l| = l$, for
	$l=0,\ldots,2k$. Following the notation introduced in 
	Subsection~\ref{upper:standarization}, we
	define four types of vectors:
	\begin{enumerate}
		\item $\ket{+_{A_k}} = x \ket{0} \ket{A_k} +
		\sum\limits_{\substack{B_{k+1}:\\|B_{k+1} \cap
				A_k| =
				k}}
		\ket{1}\ket{B_{k+1}}$,
		\item $\ket{-_{A_k}} = x \ket{1} \ket{A_k} + \sum\limits_{\substack{
				B_{k-1}: \\ |B_{k-1} \cap A_k| =
				k - 1}}
		\ket{0}\ket{B_{k-1}}$,
		\item $\ket{\oplus_{A_k}} = \sum\limits_{\substack{B_{k+1}:\\|B_{k+1}
				\cap A_k| = 1}}	\ket{1}\ket{B_{k+1}}$,
		\item $\ket{\ominus_{A_k}} = \sum\limits_{\substack{B_{k-1}:\\|B_{k-1}
				\cap
				A_k| = 0}}
		\ket{0}\ket{B_{k-1}}$,
	\end{enumerate}
	for each $A_k \subset \{1,\ldots,2k\}$ and some $x > 0$. Now we define the 
	following matrices:
	\begin{enumerate}
		\item $I_+ = \sum_{A_k} \ketbra{+_{A_k}}{A_k}$,
		\item $I_- = \sum_{A_k} \ketbra{-_{A_k}}{A_k}$,
		\item $I_\oplus = \sum_{A_k} \ketbra{\oplus_{A_k}}{A_k}$,
		\item $I_\ominus = \sum_{A_k} \ketbra{\ominus_{A_k}}{A_k}$.
	\end{enumerate}
	
	For arbitrary $A_k, A'_k \subset \{1,\ldots,2k\}$ we have
	\begin{enumerate}
		\item$\braket{+_{A_k}}{+_{A'_k}} = x^2 \delta(A_k = A'_k)
		+|\{B_{k+1}: |B_{k+1} \cap A_k| = k,|B_{k+1} \cap A'_k| = k \}|,
		$
		\item$
		\braket{-_{A_k}}{-_{A'_k}} = x^2 \delta(A_k = A'_k)
		+ |\{B_{k-1}: |B_{k-1} \cap A_k| = k-1,|B_{k-1} \cap A'_k| = k-1
		\}|,
		$
		\item$
		\braket{\oplus_{A_k}}{\oplus_{A'_k}}
		= |\{B_{k+1}: |B_{k+1} \cap A_k| = 1,|B_{k+1} \cap A'_k| = 1 \}|,
		$
		\item$
		\braket{\ominus_{A_k}}{\ominus_{A'_k}}
		=|\{B_{k-1}: |B_{k-1} \cap A_k| = 0,|B_{k-1} \cap A'_k| = 0 \}|.
		$
	\end{enumerate}
	
	We can observe that if $A_k = A'_k$, then the above inner products are $x^2 
	+
	k$, $x^2 + k$, $k$, $k$, respectively. If $|A_k \cap A'_k| = k -1$ then all 
	the inner products are
	equal to one. Finally, if $|A_k \cap A'_k| < k -1$ then we obtain all the 
	inner products are equal
	to zero. We note two useful facts about matrices $I_+, I_-, I_\oplus, 
	I_\ominus$. Firstly, we have
	\begin{equation}\label{eq:x1}
		I_+^\dagger I_+ + I_\oplus^\dagger I_\oplus = I_-^\dagger
		I_- + I_\ominus^\dagger I_\ominus.
	\end{equation}
	Secondly, one can show that
	\begin{equation}\label{eq:x2}
		\|I_+^\dagger I_+ + I_\oplus^\dagger I_\oplus\|_\infty = x^2 + 2k
		+ 2k^2.
	\end{equation}
	As far as the first equality is straightforward, to show the second one, 
	note that for each $A_k$
	there is exactly $k^2$ sets $A'_k$ such that $|A_k \cap A'_k| =k -1$. This 
	means that by the Birkhoff's
	Theorem we can express $I_+^\dagger I_+ + I_\oplus^\dagger I_\oplus$ in the 
	basis given by vectors
	$\ket{A_k}$ as $	I_+^\dagger I_+ + I_\oplus^\dagger I_\oplus = (x^2 + 
	2k) \Id + 2\sum_{i=1}^{k^2}
	\Pi_i,$ where $\Pi_i$ are permutation matrices. By the triangle inequality 
	we have that the spectral
	norm is no greater than $x^2 + 2k + 2k^2$. By taking the normalized  vector 
	$ \ket{x} \propto \sum_{A_k}
	\ket{A_k}$ we get $ \bra{x} \left(I_+^\dagger I_+ + I_\oplus^\dagger 
	I_\oplus\right) \ket{x} = x^2 +
	2k + 2k^2.$
	
	To state the upper bound for $\|H_{R,k}\|_\infty$ we will use the 
	definition of $H_R$ from
	Eq.~\eqref{fidelity-upper-proof-4} and the decomposition from 
	Eq.~\eqref{fidelity-upper-proof-5}.
	For a given $A_k, A_k' \subset \{1,\ldots,2k\}$ we have that
	\begin{equation}
		\begin{split}
			(\bra{A_k} \otimes \Id_a) H_{R,k} (\ket{A'_k} \otimes \Id_a) =
			\sum_{\substack{\pi:\\ \pi(A_k) = A'_k}} M_\pi - 
			\sum_{\substack{\pi:\\
					\pi(0,A_k) = 0,A'_k}} M_\pi = \sum_{\substack{\pi:\\ \pi(0) 
					\neq 0, \\
					\pi(A_k) = A'_k}} M_\pi -
			\sum_{\substack{\pi:\\ \pi(0) \neq 0, \\
					\pi(0,A_k) = 0,A'_k}} M_\pi.
		\end{split}
	\end{equation}
	Let us now define
	\begin{equation}\label{eq:x3}
		\begin{split}
			G_{R,k} &= (I_+^\dagger \otimes \Id_a) R (I_+ \otimes \Id_a) +
			(I_\oplus^\dagger \otimes \Id_a) R ( I_\oplus \otimes \Id_a)- 
			(I_-^\dagger
			\otimes \Id_a) R (I_- \otimes \Id_a) -
			(I_\ominus^\dagger \otimes \Id_a) R ( I_\ominus \otimes \Id_a).
		\end{split}
	\end{equation}
	
	Taking $A_k, A_k' \subset \{1,\ldots,2k\}$ we have:
	{\fontsize{10pt}{0}
	\begin{equation}
		\begin{split}
			&(\bra{A_k} \otimes \Id_a) G_{R,k} (\ket{A'_k} \otimes \Id_a) \\
			=& \left( x^2 \sum_{\substack{\pi:\\ \pi(A_k) = A'_k}} M_\pi
			+ x \sum_{\substack{B'_{k+1}, \pi:\\|B'_{k+1} \cap A'_k| = k, \\ 
			\pi(0,
					A_k) = B'_{k+1}}} M_\pi +x \sum_{\substack{B_{k+1}, \pi:\\ 
					|B_{k+1} \cap
					A_k| = k, \\\pi(B_{k+1}) = 0, A'_k}} M_\pi + 
					\sum_{\substack{ B_{k+1},
					B'_{k+1}, \pi : \\ |B_{k+1} \cap A_k|
					= k,\\ |B'_{k+1} \cap A'_k| = k,\\ \pi(0,B_{k+1}) = 0, 
					B'_{k+1}}}
			M_\pi+ \sum_{\substack{B_{k+1}, B'_{k+1}, \pi: \\ |B_{k+1} \cap
					A_k| = 1, \\ |B'_{k+1} \cap A'_k| = 1,\\ \pi(0,B_{k+1}) =
					0,B'_{k+1}}} M_\pi \right) \\-&\left( x^2 \sum_{\substack{ 
					\pi: \\
					\pi(0,A_k) = 0,A'_k}}
			M_\pi + x
			\sum_{\substack{ B'_{k-1}, \pi: \\ |B'_{k-1} \cap A'_k| = k-1,\\
					\pi(A_k) = 0,B'_{k-1}}} M_\pi + x
			\sum_{\substack{B_{k-1}, \pi: \\ |B_{k-1} \cap A_k| = k-1, \\
					\pi(0,B_{k-1}) = A'_k}} M_\pi + \sum_{\substack{B_{k-1}, 
					B'_{k-1}, \pi: \\
					|B_{k-1} \cap A_k| = k-1, \\ |B'_{k-1} \cap A'_k| = k-1, \\ 
					\pi(B_{k-1}) =
					B'_{k-1}}} M_\pi + \sum_{\substack{B_{k-1}, B'_{k-1}, \pi: 
					\\
					|B_{k-1} \cap A_k| = 0, \\ |B'_{k-1} \cap A'_k| = 0, \\
					\pi(B_{k-1}) = B'_{k-1}} } M_\pi \right).
		\end{split}
	\end{equation}}
	This can be simplified to
	\begin{equation}
		\begin{split}
			&(\bra{A_k} \otimes \Id_a) G_{R,k} (\ket{A'_k} \otimes \Id_a) \\
			=& \left( x^2 \sum_{\substack{\pi:\\ \pi(A_k) = A'_k}} M_\pi
			+ x \sum_{\substack{B'_{k+1}, \pi:\\|B'_{k+1} \cap A'_k| = k, \\ 
			\pi(0,
					A_k) = B'_{k+1}}} M_\pi + x \sum_{\substack{B_{k+1}, \pi:\\
					|B_{k+1} \cap A_k| = k,
					\\
					\pi(B_{k+1}) = 0, A'_k}} M_\pi \right)\\
			-&\left( x^2 \sum_{\substack{
					\pi: \\ \pi(0,A_k) = 0,A'_k}} M_\pi + x \sum_{\substack{ 
					B'_{k-1}, \pi: \\
					|B'_{k-1} \cap A'_k| =
					k-1,\\
					\pi(A_k) = 0,B'_{k-1}}} M_\pi + x
			\sum_{\substack{B_{k-1}, \pi: \\ |B_{k-1} \cap A_k| = k-1, \\
					\pi(0,B_{k-1}) = A'_k}} M_\pi \right).
		\end{split}
	\end{equation}
	
	Let us write the above as $(\bra{A_k} \otimes \Id_a) G_{R,k} (\ket{A'_k}
	\otimes \Id_a) = \sum_\pi c_\pi M_\pi$, where $c_\pi$ are some constants. 
	For
	each $\pi$, let us determine the value of $c_\pi$:
	
	\begin{itemize}
		\item For $\pi$ such that $\pi(0) = 0, \pi(A_k) = A_k'$ we have $c_\pi
		=x^2 - x^2 = 0$.
		\item For $\pi$ such that $\pi(0) = 0, \pi(A_k) \neq A_k'$ we have
		$c_\pi = 0$.
		\item For $\pi$ such that $\pi(0) \neq 0, \pi(A_k) = A_k'$ we have
		$c_\pi = x^2 + x + x = x^2 + 2x$.
		\item For $\pi$ such that $\pi(0) \neq 0, \pi(A_k) \neq A_k', \pi(0,
		A_k) \neq 0, A_k'$ there exists $a_0 \not\in \{0\} \cup A_k
		$, such that $\pi(a_0) \in \{0\} \cup A'_k $. Therefore, we
		consider two sub-cases:
		\begin{itemize}
			\item If for each $a \not\in \{0\} \cup A_k$ it holds $\pi(a)
			\not\in A'_k$, then $\pi(a_0) = 0$, $\pi(0) \in A'_k$ and $A'_k
			\subset \pi(0,A_k)$. Then, $c_\pi = x - x = 0$.
			\item If $\pi(a_0) \in A'_k$, then we have two options:
			\begin{itemize}
				\item If $\pi(a_0, A_k) = 0,A'_k$, then $c_\pi = x - x =0$.
				\item If $\pi(a_0, A_k) \neq 0,A'_k$, then $c_\pi = 0$.
			\end{itemize}
		\end{itemize}
		\item For $\pi$ such that $\pi(0) \neq 0, \pi(0, A_k) = 0, A_k'$ we
		have $c_\pi = -x^2- x - x = -(x^2 + 2x)$.
	\end{itemize}
	Therefore, we can see that $G_{R,k} = (x^2 + 2x)H_{R,k}$. Then, utilizing 
	Eq.~\eqref{eq:x1},
	Eq.~\eqref{eq:x2} and Eq.~\eqref{eq:x3} we get
	\begin{equation}
		-( x^2 + 2k + 2k^2) \Id \leq G_{R,k} \leq ( x^2 + 2k + 2k^2) \Id,
	\end{equation}
	and finally we obtain $ \|H_{R,k}\|_\infty \le \frac{x^2 + 2k + 2k^2}{x^2 + 
		2x}.$ Minimizing over $x > 0$, we get for $x \approx 2k^2$ that 
	$\|H_{R,k}\|_\infty \le 1 - \Theta(1/k^2)$, which
	finishes this case of the proof.
\end{proof}

\subsection{Proof of Lemma~\ref{lem-upper-main}}\label{ss:gen-n}

\begin{proof}[Proof of Lemma~\ref{lem-upper-main}]
	We have the following sequence of conclusions
	\begin{equation*}
		\begin{array}{rlcl}
			F_2 & \le \widetilde{F_2} && \mbox{(by 
			Eq.~\eqref{eq:fidelity-appendix},
				Eq.~\eqref{eq:fidelity-new-appendix},
				Lemma~\ref{lemma-unitary-transformation})}\\
			& =  \frac12 + \frac12 \max\limits_{k = 0,\ldots, \floor{N/ 2}}
			\max\limits_{\substack{R: \,\, 0 \le R \le \Id\\ R = \sum_{\pi}
					W_\pi^{\top_\ZZ}
					\otimes M_\pi }}  \|H_{R,k}\|_\infty && \mbox{(by
				Lemma~\ref{lemma-simply})}\\
			& \le  \frac12 + \frac12 \max\limits_{k =
				0,\ldots, \floor{N/ 2}} \max\limits_{\substack{\widetilde R: 
				\,\, 0 \le
					\widetilde R \le
					\Id\\ 
					\widetilde R = \sum_{\pi} \widetilde W_\pi^{\top_\ZZ} 
					\otimes
					\widetilde M_\pi
			}} \|\widetilde H_{\widetilde R,N-k}\|_\infty && \mbox{(by
				Lemma~\ref{lemma-last})}\\
			& \le  \frac12 + \frac12 \max\limits_{k =
				0,\ldots, \floor{N/ 2}} 1 - \Theta\left(\frac{1}{(N-k)^2}\right)
			&& \mbox{(by
				Lemma~\ref{lemma-k-2k})}\\
			& =  1 - \Theta\left(\frac{1}{N^2}\right).
			&&
		\end{array}
	\end{equation*}
\end{proof}

\section{Proof of upper bound}\label{app:upper}
\begin{lemma}\label{lemma-upper-d}
	The maximum value of the average fidelity 
	function, defined in
	Eq.~\eqref{eq:fidelity} is upper bounded by
	\begin{equation}
		F_d \le 1 - \Theta\left(\frac{1}{N^2}\right).
	\end{equation}
\end{lemma}

\begin{proof}
	The thesis is true for $d=2$ which follows from Lemma~\ref{lem-upper-main}. 
	Let us fix $d \in \N$. Take the optimal learning scheme $\LL$ such 
	that it achieves $\mathcal{F}_d^{\text{avg}}(\LL) = F_d$. Without loss of 
	the generality we assume that $\LL$ satisfies
	\begin{equation}
		[L, \Id_d \otimes U \otimes (\Id_d \otimes \bar{U})]^{\otimes N}]=0
	\end{equation}
	for any unitary matrix $U \in \mathrm{M}(\HH_d)$.
	Then, for any $U$, we have
	\begin{equation}
		F_d = \mathcal{F}_d(\PP_U, \QQ_U) \le \frac{d-1}{d} + \frac1d \min_i 
		\tr(P_{U, i} Q_{U, i}).
	\end{equation}
	Moreover, for $j_1 \neq j_2$ it holds
	\begin{equation}
		\tr(P_{U, j_1} Q_{U, j_2}) \le \tr(P_{U, j_1} (\Id_d - Q_{U, j_1})) \le 
		1 - \min_i 
		\tr(P_{U, i} Q_{U, i}).
	\end{equation}
	
	Now, we use $\LL$ to construct a new learning scheme $\LL'$ of qubit von 
	Neumann measurements in the following way. 
	
	Let $\Pi \in \mathrm{M}(\HH_d)$ 
	be a projector onto 
	$\ket{0}, \ket{1}$ and define an isometry matrix $V = \sum_{ i=0}^1 
	\proj{i} \in \mathrm{M}(\HH_d,\HH_2)$. Having access to unknown  qubit von 
	Neumann measurement $\PP_U$ we may use it to implement  von 
	Neumann measurement $\PP_{U \oplus \Id_{d-2}}$ acting on $\Omega(\HH_d)$. 
	To do that, we take 
	$\sigma \in \Omega(\HH_d)$ and measure it in the following way
	\begin{equation}
		\proj{0} \otimes \Pi \sigma \Pi + \proj{1} \otimes \sum_{i=2}^{d-1} 
		\tr(\sigma \proj{i}) \proj{i}.
	\end{equation}
	If on the first system we measure ``$0$'', then we can project the state 
	$V^\dagger (\Pi \sigma \Pi) V$ into $\Omega(\HH_2)$ and measure it by using 
	$\PP_U$. Otherwise, we do nothing. As a result we implemented the 
	measurement of the form
	\begin{equation}
		\sum_{i=0}^1 \tr(V^\dagger \sigma V U \proj{i} U^\dagger) \proj{i} + 
		\sum_{i=2}^{d-1} \tr(\sigma \proj{i}) \proj{i} = \PP_{U \oplus 
		\Id_{d-2}}(\sigma).
	\end{equation}
	During the retrieval stage, we project the input state $\rho \in 
	\Omega(\HH_2)$ 
	into $V\rho V^\dagger$. Moreover, as the output of $\LL$ is a classical 
	label ``$0$'', $\ldots$, ``$d-1$'', to finalize the construction of $\LL'$, 
	all the labels ``$1$'', $\ldots$, ``$d-1$'' are returned as ``$1$''.
	
	For a given unitary matrix $U \in \mathrm{M}(\HH_2)$ and the learning 
	network $\LL'$ we may calculate
	\begin{equation}
	\begin{split}
		\mathcal{F}_2(\PP_U, \QQ_U) &= \frac{1}{2} \left(\tr(VP_{U,0}V^\dagger 
		Q_{U \oplus \Id_{d-2}, 0 } ) + \tr(VP_{U,1}V^\dagger(\Id_d -  
		Q_{U \oplus \Id_{d-2}, 0 }) )\right)\\
		&\ge \frac12 \min_i 
		\tr(P_{U \oplus \Id_{d-2}, i} Q_{U \oplus \Id_{d-2}, i}) + \frac12 
		\left(1 - \tr(P_{U \oplus \Id_{d-2}, 1} Q_{U \oplus \Id_{d-2}, 
		0})\right)\\
	&\ge \min_i 
	\tr(P_{U \oplus \Id_{d-2}, i} Q_{U \oplus \Id_{d-2}, i}) \ge d F_d - (d-1).
	\end{split}
	\end{equation}
	Therefore, we get
	\begin{equation}
		dF_d - (d-1) \le \mathcal{F}_2^{\text{avg}}(\LL') \le F_2 \le 1 - 
		\Theta\left(\frac{1}{N^2}\right),
	\end{equation}
which ends the proof.
\end{proof}

\section{Pretty good learning scheme}\label{app:pgls}

The pretty good learning scheme $\LL_{PGLS}= \left( \sigma, \{ \CC_i 
\}_{i=1}^{N-1}, \RR \right)$ consists
of the initial state $\sigma$, which is a tensor product of $N$ copies of the 
maximally entangled
state $\ket{ \omega} = \frac{1}{\sqrt{2}} \ketv{\Id_2}$, processing channels 
$\{
\CC_i\}_{i=1}^{N-1}$ that are responsible for majority voting (see 
Section~\ref{qubit-pgls}) and
the measurement $\RR = \{ R, \Id - R \}$. To construct the effect $R$, we fix 
$N_0 \in \N$ and take
$n= N_0-1$. Let us define
\begin{equation}
	s_n(k,m) \coloneqq \sum_{i=0}^k \sum_{j=0}^{n-k} \delta_{i+j-m} {k \choose
		i } { n-k \choose j } (-1)^{n-k-j },
\end{equation}
being the convolution of binomial coefficients. We consider the effect $R$ of 
the form
\begin{equation}\label{effect-rr}
	\begin{split}
		R &= \sum_{k=0}^{n} \proj{R_k},  \quad \text{ such that}
		\ket{R_k} = \frac{\ketv{M_k}}{||M_k||_2},  \\ 
		\mathrm{M}(\HH_2,
		\HH_{2^{n+1}} ) \ni M_k &= \sum_{m = 0}^{n+1} \frac{s_n(k, n - 
			m)\ket{0} + 
			s_n(k, n + 1 -
			m)\ket{1}}{\sqrt{n + 1 \choose m}} \bra{ D_m^{n+1}},
	\end{split}
\end{equation}
for $k = 0, \ldots, n$. 
\begin{lemma}\label{remark-ab}
	Let $\ket{x} = \left[\begin{array}{cc}a\\b\end{array}\right]$, $a,b
	\in \C$. Then, we have $M_k \ket{x}^{\otimes n+1} = (a+b)^k (a-b)^{n-k}
	\ket{x}$.
\end{lemma}
\begin{proof}
	Direct calculations reveal
	\begin{equation}
		\begin{split}
			M_k \ket{x}^{\otimes n+1} & =\left[
			\begin{array}{cc} \sum_{m=0}^n {n+1 \choose n-m} \cdot
				\frac{s_n(k,m)}{{n+1 \choose n-m}} a^{m+1} b^{n-m} \\ 
				\sum_{m=0}^n {n+1
					\choose n+1-m} \cdot \frac{s_n(k,m)}{{n+1 \choose n+1-m}} 
				a^{m}
				b^{n+1-m} \end{array} \right]  = \sum_{m=0}^n s_n(k,m) a^m 
			b^{n-m}
			\ket{x} \\ & = (a+b)^k (a-b)^{n-k} \ket{x}.
		\end{split}
	\end{equation}
\end{proof}
To prove that $R$ is a valid effect, let us now define \begin{equation}
	M \coloneqq \left[ s_n(k,m) \right] _{k,m=0}^n
\end{equation} and
a diagonal matrix \begin{equation}
	D \coloneqq \sum_{m=0}^n  \frac{1}{{n \choose m}} \proj{m}. 
\end{equation} 
\begin{lemma}\label{remark-squere}
	With the notation given above, it holds that $M^2 = 2^n \Id_{n+1}$.
\end{lemma}

\begin{proof}
	First, observe that $\C^{n+1} = \text{span} \left( [x^k]_{k=0}^n: x \in \C
	\right).$ Let us take any vector of the form $\ket{x} \coloneqq
	[x^k]_{k=0}^n$, where $ x \in \C$. We have
	\begin{equation}
		\begin{split}
			M\ket{x} &= \left[ \sum_{m=0}^n s_n(k,m)x^m \right]_{k=0}^n =
			\left[ (x+1)^k (x-1)^{n-k} \right]_{k=0}^n = (x-1)^n \left[
			\left(\frac{x+1}{x-1}\right)^k \right]_{k=0}^n.
		\end{split}
	\end{equation}
	Finally, we calculate
	\begin{equation}
		\begin{split}
			M^2\ket{x} &= (x-1)^n \left( \frac{x+1}{x-1} -1\right)^n \left[ 
			\left(
			\frac{\frac{x+1}{x-1} + 1}{\frac{x+1}{x-1} - 1} \right)^k
			\right]_{k=0}^n = 2^n \ket{x}.
		\end{split}
	\end{equation}
	
\end{proof}

\begin{lemma}\label{remark-symetric}
	Using the notation presented above, we have the following equation $MD = 
	(MD)^\top $.
\end{lemma}

\begin{proof}
	We will show that $\bra{k} MD \ket{m} = \bra{m} MD \ket{k}$ for any $m,k =
	0,\ldots,n$. W.l.o.g. we can assume that $k < m$. On the one hand, it holds
	that
	\begin{equation}\label{proof-MD-1}
		\begin{split}
			\bra{k} MD \ket{m} &= \frac{s_n(k,m)}{\binom{n}{m}} = 
			\sum_{\substack{
					i = 0,\ldots,k \\
					j = 0, \ldots, n - k \\
					i + j = m}}
			\frac{(-1)^{n-k-j} \binom{k}{i} \binom{n-k}{j}}{\binom{n}{m}}= 
			\sum_{i = \max \left(0, m
				+ k - n\right) }^k
			\frac{(-1)^{n-k-m+i} \binom{k}{i} \binom{n-k}{m - 
					i}}{\binom{n}{m}}\\
			&= (-1)^{n-k-m} \sum_{i = \max \left(0, m
				+ k - n\right) }^k (-1)^{i}
			\frac{k! m! (n-k)! (n-m)!}{n!i!(k-i)!(m-i)!(n-k-m+i)!}.
		\end{split}
	\end{equation}
	On the other hand, we can calculate
	\begin{equation}\label{proof-MD-2}
		\begin{split}
			\bra{m} MD \ket{k}& = \frac{s_n(m,k)}{\binom{n}{k}} = 
			\sum_{\substack{
					i = 0,\ldots,m \\
					j = 0, \ldots, n - m \\
					i + j = k}}
			\frac{(-1)^{n-m-j} \binom{m}{i} \binom{n-m}{j}}{\binom{n}{k}}= 
			\sum_{i = \max \left(0, m
				+ k - n\right) }^k
			\frac{(-1)^{n-k-m+i} \binom{m}{i} \binom{n-m}{k - 
					i}}{\binom{n}{k}}\\
			&= (-1)^{n-k-m}\sum_{i = \max \left(0, m
				+ k - n\right) }^k (-1)^{i}
			\frac{k! m! (n-k)! (n-m)!}{n!i!(k-i)!(m-i)!(n-k-m+i)!},
		\end{split}
	\end{equation}
	which gives us the desired equality and  completes the proof.
\end{proof}

\begin{lemma}\label{measurement-correctly}
	The operator $R$ defined in Eq.~\eqref{effect-rr} satisfies $0 \le R \le 
	\Id_{2^{n+2}}$ and 
	therefore $\RR = \{ R, \Id - R\}$ is a valid POVM. 
\end{lemma}

\begin{proof}
	Let us fix $N_0 \in \N$ and take $n= N_0-1$. Let us consider a matrix $X 
	\coloneqq \frac{n+2}{n+1}
	MDM^\top$. On the one hand, by using Lemma~\ref{remark-squere} and 
	Lemma~\ref{remark-symetric},  we
	get
	\begin{equation}
		X = \frac{n+2}{n+1 } \left(MD\right)^\top M^\top = \frac{n+2}{n+1}
		D(M^2)^\top = \frac{n+2}{n+1} 2^n D.
	\end{equation}
	On the other hand, we have
	\begin{equation}
		\begin{split}
			\tr\left(M_k^\dagger M_{k'}\right)  &= \sum_{m=0}^n \frac{s_n(k,m)
				s_n(k',m)}{{n+1 \choose n-m}}  +  \sum_{m=0}^n 
			\frac{s_n(k,m)s_n(k',m)}{{n+1
					\choose	n+1-m}} \\&= \sum_{m=0}^n 	s_n(k,m) s_n(k',m) 
			\left[ 
			\frac{1}{{n+1
					\choose n-m}} + \frac{1}{{n+1 	\choose n-m+1}} \right]  
			\\&=
			\frac{n+2}{n+1} \sum_{m=0}^n \frac{s_n(k,m)
				s_n(k',m)}{{n \choose m}}  = \bra{k}X \ket{k'}.
		\end{split}
	\end{equation}
	Therefore, for all $k \not = k'$ we get $ \tr\left(M_k^\dagger 
	M_{k'}\right) = 
	0 $. According to the definition Eq.~\eqref{effect-rr}, we get 
	$\braket{R_k}{R_{k'}}=\delta_{k,k'}$, which gives us $0 \le R
	\le \Id_{2^{n+2}}$.
	
\end{proof}

\begin{lemma}\label{form-of-q}
	
	Let us fix $N_0 \in \N$. The
	approximation $\QQ_U = \{ Q_{U,0}, \Id_2 - Q_{U,0}\}$ of the von Neumann
	measurement $\PP_{U}  $ obtained in the pretty good learning scheme is of
	the form
	\begin{equation}
		\begin{split}
			Q_{U,0} = \frac{N_0}{N_0+1} P_{U,0}.
		\end{split}
	\end{equation}
\end{lemma}
\begin{proof}
	Given a unitary matrix $U$ we take
	$P_{U,0} = \proj{x}$ for some unit vector $\ket{x} \in \HH_2$. Let us
	decompose
	the $(n+2)$-qubit space in the following way $\HH_{2^{n+2}} = \ZZ \otimes
	\XX$, where $\ZZ = \HH_2$ and $\XX = \HH_{2^{n+1}}$. In the proof of
	Lemma~\ref{measurement-correctly} we defined the matrix $X =
	\frac{n+2}{n+1} MDM^\top$ and showed that $X = \frac{n+2}{n+1} 2^n D$, and
	$\tr\left(M_k^\dagger M_{k'}\right)  = \bra{k}X \ket{k'}$. Therefore, for
	any $k= 0,\ldots,n$ we have $ \| M_k \|_2^2 = \frac{n+2}{n+1}
	\frac{2^n}{{n \choose k}} $. Due to this fact and by
	Lemma~\ref{remark-ab}, we may express the effect $Q_{U, 0}$ as
	\begin{equation}
		\begin{split}
			Q_{U,0} &=	\tr_{\XX} \left( \left( \Id_2 \otimes
			\overline{P_{U,0}}^{\otimes n + 1} \right) R \right) = \left(\Id_2
			\otimes\overline{\bra{x}}^{\otimes n+1} \right) R \left( \Id_2
			\otimes\overline{\ket{x}}^{\otimes n+1} \right)  = \sum_{k=0}^{n}
			\frac{1}{\| M_k \|_2^2} M_k
			\proj{x}^{\otimes n+1} M_k^\dagger \\ &= \sum_{k=0}^n
			\frac{1}{\| M_k \|_2^2} |a+b|^{2k}
			|a-b|^{2(n-k)} \ketbra{x}{x} = \frac{n+1}{n+2} \sum_{k=0}^n 
			\frac{{n\choose
					k}}{2^n} |a+b|^{2k} |a-b|^{2(n-k)} \ketbra{x}{x} \\ &= 
			\frac{n+1}{n+2}
			\frac{(|a+b|^2+|a-b|^2)^n}{2^n} \ketbra{x}{x} = \frac{n+1}{n+2}
			\ketbra{x}{x},
		\end{split}
	\end{equation}
	which completes the proof.
\end{proof}

\end{document}